\documentclass[11pt]{article}
\usepackage[en-US]{datetime2}
\usepackage{color,array,graphics}
\usepackage{enumerate}
\usepackage[papersize={8.5in,11in},top=1in,left=1in,right=1in,bottom=1in]{geometry}

\usepackage{amsmath}
\usepackage{amsthm}
\usepackage{amssymb}
\usepackage{mathrsfs}
\usepackage{multirow}
\usepackage{tikz}
\usepackage{sectsty}
\usepackage{float}

\newtheorem{theorem}{Theorem}[section]

\newtheorem{lemma}{Lemma}[theorem]
\newtheorem{definition}{Definition}[section]

\newtheorem{algorithm}{Algorithm}

\newcommand{\cc}{w} 
\newcommand{\dc}{w} 
\newcommand{\rc}{c} 
\newcommand{\tc}{\tilde{c}} 
\newcommand{\BR}{BR} 
\newcommand{\nBR}{\overline{BR}} 
\newcommand{\tPhi}{\widetilde{\Phi}} 
\newcommand{\CC}[2]{CC_{#1}({#2})} 
\newcommand{\DC}[3]{DC_{#1}({#2, #3})} 
\newcommand{\Q}{Q} 

\begin{document}

\title{Forming better stable solutions in Group Formation Games inspired by Internet Exchange Points (IXPs)}


\author{Elliot Anshelevich \and Wennan Zhu}

\date{{\small Rensselaer Polytechnic Institute, Troy, NY\\ \today}}

\maketitle{}

\begin{abstract} We study a coordination game motivated by the formation of Internet Exchange Points (IXPs), in which agents choose which facilities to join. Joining the same facility as other agents you communicate with has benefits, but different facilities have different costs for each agent. Thus, the players wish to join the same facilities as their ``friends", but this is balanced by them not wanting to pay the cost of joining a facility.
We first show that the Price of Stability ($PoS$) of this game is at most 2, and more generally there always exists an $\alpha$-approximate equilibrium with cost at most $\frac{2}{\alpha}$ of optimum. We then focus on how better stable solutions can be formed. If we allow agents to pay their neighbors to prevent them from deviating (i.e., a player $i$ {\em voluntarily} pays another player $j$ so that $j$ joins the same facility), then we provide a payment scheme which stabilizes the solution with minimum social cost $s^*$, i.e. PoS is 1. In our main technical result, we consider how much a central coordinator would have to pay the players in order to form good stable solutions. Let $\Delta$ denote the total amount of payments needed to be paid to the players in order to stabilize $s^*$, i.e., these are payments that a player would lose if they changed their strategy from the one in $s^*$. We prove that there is a tradeoff between $\Delta$ and the Price of Stability: $\frac{\Delta}{cost(s^*)} \le 1 - \frac{2}{5} PoS$. Thus when there are no good stable solutions, only a small amount of extra payment is needed to stabilize $s^*$; and when good stable solutions already exist (i.e., $PoS$ is small), then we should be happy with those solutions instead. Finally, we consider the computational complexity of finding the optimum solution $s^*$, and design a polynomial time $O(\log n)$ approximation algorithm for this problem.
\end{abstract}

\section{Introduction}

We study a coordination game motivated by the formation of Internet Exchange Points (IXPs). In this game, there are $m$ facilities available, and the players (modeling ISPs, or more generally entities which wish to exchange traffic with each other) choose which facilities to join. Joining a facility $f_k$ has a cost for player $i$, which we call the ``connection cost" and denote by $\cc(i, f_k)$; this cost can be different for different players and facilities. The reason why players are willing to pay such costs is because joining the same facility as other players is beneficial: a pair of players $i$ and $j$ which do not connect to the same facility must pay a cost $\dc(i,j)$, but if they share a facility then this cost disappears. Finally, the facilities themselves have costs $c(f_k)$ which must be paid for by the players using these facilities. In summary, the players wish to join the same facilities as their ``friends" in order to avoid paying the costs $\dc(i,j)$, but this is counterbalanced by them not wanting to pay the cost of joining a facility.

While our game is quite general, and models general group formation (e.g., facilities are clubs or groups people can join, and they wish to join the same clubs as their friends), this game is specifically inspired by the formation of IXPs in the Internet. IXPs are facilities where Internet Service Providers (ISPs) can exchange Internet traffic with high speed; a large fraction of total Internet traffic flows through such hubs \cite{ager2012anatomy}. If two ISPs join the same IXP (and pay their cost for joining, which can depend on many factors including the pricing scheme and the physical location(s) of the IXP), then they gain the benefit of mutual high speed communication.
If, however, two ISPs do not use the same IXP, they must use alternate means of exchanging traffic with each other (e.g., through their providers or private peering), which we model by them incurring an extra cost $\dc(i,j)$.

Coordination games have been widely studied in various situations where agents gain utility by forming coalitions with other agents. Even with the large amount of existing work on both coordination games and group formation, the questions we consider in this paper have not been studied before for our game (see Related Work). Like many such games, ours can be represented by a graph, in which each node stands for a player and the edges between them have weights representing the disconnection cost for them not belonging to the same facility. One major difference between our game and much (although certainly not all) of existing work is that the facilities (i.e., groups that players can join) are not identical: their quality for a player $i$ depends not only on who else has joined the same group (as in hedonic games \cite{hedonic}), but also on the specific facility being joined, as quantified by the cost $\cc(i,f_k)$. This immediately changes a lot about equilibrium structure: it is no longer the case that everyone being in the same group is an equilibrium solution which minimizes social cost; instead equilibrium solutions involve players balancing their cost for joining facilities with their cost of being separated from their friends. Other coordination games look at cases where only a limited number of facilities can be open, or when players have both ``friends" and ``enemies" (i.e., $\dc(i,j)$ can be negative); for the types of settings we consider, however, all facilities can open as long as players are able to pay for them, and there is never any additional cost from two players joining the same facility (i.e., $\dc(i,j)\geq 0$). Moreover, unlike most other coordination games, we assume that facilities have a cost which must be shared among the players using it, which adds a significant layer of complexity to our results (for example, our game is no longer a potential game \cite{monderer1996potential}). For more details and comparison with existing work, see the Related Work section.

\subsection*{Our Contributions}
In this paper, we study a coordination game where a strategy of an agent $i$ is to choose a facility $f_k$ to join, by paying a connection cost $\cc(i, f_k)$ (or to not join any facility). If two agents $i, j$ do not use the same facility, then both of them are charged a disconnection cost $\dc(i, j)$. In addition, there is a fixed facility cost $c(f_k)$ for each open facility, which is split among all agents using $f_k$ according to an arbitrary pricing rule. An agent's total social cost is the sum of its connection cost, disconnection cost, and its share of the facility cost. An assignment with a pricing rule is stable if it is budget balanced (each $c(f_k)$ is fully paid by all agents using $f_k$), and no agent wants to switch facilities, i.e., it is a Nash equilibrium.


We study the quality of equilibrium solutions for this game, as well as ways to create new stable solutions. We first show that while the Price of Anarchy can be arbitrarily high, the Price of Stability ($PoS$) is at most 2, and more generally there always exists an $\alpha$-approximate equilibrium with cost at most $\frac{2}{\alpha}$ of optimum. While we use potential arguments to prove this \cite{tardos2007network}, note that this game is not a potential game due to facility costs, and thus new proof techniques are needed beyond simply defining a potential function. We then focus on how better stable solutions can be formed. If we allow agents to pay their neighbors to prevent them from deviating (i.e., a player $i$ {\em voluntarily} pays another player $j$ so that $j$ joins the same facility), then we provide a payment scheme which stabilizes the solution with minimum social cost $s^*$, i.e. PoS is 1. This is essentially what occurs, for example, in paid peering \cite{shrimali2006paid}, where two ISPs have different incentives, and so one ISP pays the other in order to form a peering connection. Finally, for our main result, we consider how much a central coordinator would have to pay the players in order to form good stable solutions, similarly to \cite{anshelevich2014approximate, bachrach2009cost}. Let $\Delta$ denote the total amount of payments needed to be paid to the players in order to stabilize $s^*$, i.e., these are payments that a player would lose if they changed their strategy from the one in $s^*$. We prove that there is a tradeoff between $\Delta$ and the Price of Stability: $\frac{\Delta}{cost(s^*)} \le 1 - \frac{2}{5} PoS$. [See Figure~\ref{fig:delta}]. Thus when there are no good stable solutions, only a small amount of extra payment is needed to stabilize $s^*$; and when good stable solutions already exist (i.e., $PoS$ is small), then we should be happy with those solutions instead! This result is proven by forming several solutions where specific subsets of players perform their best responses, and then showing that when a small amount of payment is not enough to stabilize $s^*$, then at least one of these solutions is guaranteed to be better than $s^*$, giving a contradiction. The difficulty here results from the fact that letting any single player move to their best response strategy from $s^*$ could still result in solutions worse than $s^*$; to get a contradiction and form a solution strictly better than $s^*$ requires changing the strategy of many players simultaneously.

\begin{figure}[htb]
\begin{center}
\includegraphics[scale = 0.5]{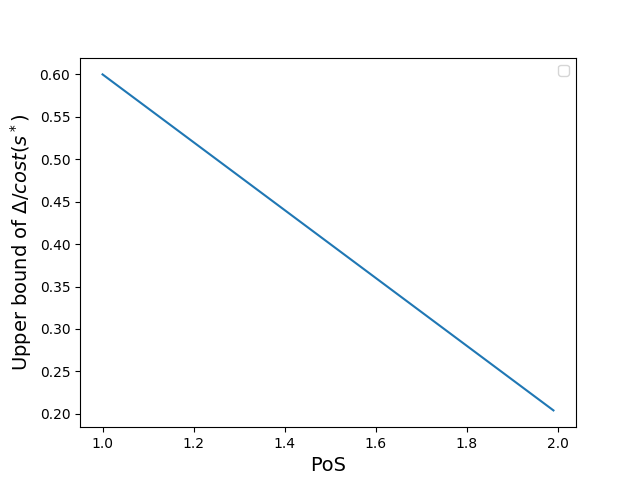}\\
\end{center}
\caption{Tradeoff between $\Delta$ and PoS: $\frac{\Delta}{cost(s^*)} \le 1 - \frac{2}{5} PoS$. }
\label{fig:delta}
\end{figure}

The results above are for the setting where each agent can join at most one facility at a time. In Section \ref{sec-multi}, we study the setting where each agent is allowed to use multiple facilities simultaneously. Many of the results above still hold for this general mode, but only under the assumption that a player can only switch their strategy by leaving one facility at a time (although it is allowed to join multiple new facilities at once).

Finally, we consider the computational complexity of finding the optimum solution $s^*$. We prove that computing it is NP-Hard (and in fact inapproximable to better than $\Omega(\log n)$ unless P=NP), and design a polynomial time approximation algorithm that gives a $\min\{m+1, O(\log n)\}$-approximation to the optimal solution (with $n$ being the number of players, and $m$ the number of facilities). We also provide a simple 2-approximation algorithm when all facility costs are zero.

\section{Related Work}There is a very large amount of work on both group formation and coordination games, which is too large to survey here.
Hedonic games \cite{hedonic, dreze1980hedonic} is an important class of games related to coordination games, in which the agents form groups, and each agent's utility only depends on the other agents in its own group, but is not affected by how agents are arranged in other groups. The objectives are usually maximizing social welfare \cite{apt2014coordination, aziz2019fractional, aziz2013computing, branzei2009coalitional, feldman2015hedonic, gairing2010computing} or minimizing social cost \cite{feldman2015hedonic}. Often, although not always, all players in a group have the same cost or utility. In much of the work, the number of groups is fixed \cite{bhalgat2010approximating, feldman2015hedonic, gourves2010max, hoefer2007cost}. 
There are also various utility/cost functions which have been studied, with the most common one being that an agent's utility is the total utility gained from being with all other agents in its group. In fractional hedonic games \cite{aziz2019fractional, bilo2014nash, bilo2015price}, an agent's utility is the average value of its presence to every other agent in the group.
More generally, there are also other types of related group formation games, e.g., congestion games \cite{christodoulou2005price, rosenthal1973class} and profit sharing games (see \cite{augustine2011dynamics} and references therein), where an agent's utility only depends on the size of the group.

 While coordination games can be considered a special case of general hedonic games, usually coordination games involve players with some sort of graph structure, where for a pair of players, being in the same group gives them both a benefit if they are ``friends" (or a penalty if they are ``enemies"). This is in contrast to many hedonic games, where all players in a group have the same utility, or the total utility of a group is somehow shared among its participants. In most related work, either the objective functions of the players are very different from ours (e.g., they depend on the number of players in their group) \cite{apt2014coordination, aziz2019fractional, branzei2011social}, or there are players who specifically don't want to be in the same group (``enemies", negative-weight edges) \cite{auletta2016generalized, bhalgat2010approximating, feldman2015unified}, or all groups are identical and the optimum solution would correspond to either everyone joining the same group or everyone forming a group on their own \cite{apt2014coordination, aziz2013computing, bhalgat2010approximating, branzei2009coalitional, feldman2015hedonic, feldman2015unified}. In contrast, our work is motivated by settings where everyone would like to form one group together to reduce the disconnection cost, but the complexity in the solution structure comes from the players trading this desire off with their individual connection costs to (non-identical) facilities.

As discussed in the Introduction, general coordination games include other settings in which an agent's utility or cost also depends on which group it joins (i.e., the groups are not identical). Our work is more closely related to this type of game. Using a graph representation, one can think of such games as either Max-Uncut (maximize the weight of edges to friends in your group) or Min-Cut (minimize the weight of edges to friends not in your group) objectives, but with additional utility or cost depending on which group a player joins (which can be modeled using additional ``anchor nodes" which {\em must} belong to a specific group, see e.g., \cite{anshelevich2014approximate}). Work on such coordination games with non-identical groups or facilities includes \cite{anshelevich2014approximate, auletta2017robustness,chierichetti2013discrete}.
In k-Coloring games \cite{carosi2019generalized} each agent gains utility by choosing a certain color/facility, and loses utility by choosing the same color as other adjacent agents, i.e., they are anti-coordination games in which all agents want to be in different groups if possible (see references in \cite{carosi2019generalized} for more discussion of such games). In generalized Discrete Preference Games \cite{auletta2016generalized}, there are exactly two groups, and the players could be friends or enemies. Similar to hedonic games, the research in this area usually focused on properties of stable solutions, e.g., \cite{auletta2016discrete} studies how a single agent could affect the Nash Equilibria converged from best responses, and \cite{auletta2017robustness} compares the prices of anarchy and stability under different objective functions.

Perhaps the most related work to ours is \cite{chierichetti2013discrete}, as it is also a Min-Cut game with non-identical groups. The main differences between our work and \cite{chierichetti2013discrete} are: it is assumed in \cite{chierichetti2013discrete} that every agent has a favorite group, and an agent's cost depends on the distance between its current group and favorite group, and the distances to its neighbors. We do not bind each agent's cost with a group in our setting. Suppose two agents have the same ``favorite group'' $f_k$, which is the group with the lowest connection cost to them; in our model their cost to any other group $f_k'$ could be very different. \cite{chierichetti2013discrete} also focuses on the setting where the group locations form a general metric or tree metric, while we do not have such assumptions. Last but not least, unlike in the works mentioned above, we assume there is a facility cost to open each facility, with different facilities having different costs. We study stable states where the facility cost is split among the agents using it, so each facility is paid for, and each agent is stable with three types of costs: connection cost to facility, its own share of facility cost, and disconnection cost to its neighbors that use different facilities. We also study the case that each agent can join multiple groups. 
Finally, parts of our work are also closely related to \cite{anshelevich2014approximate}, which shows that an optimal solution could be stabilized by providing a reasonable amount of payments to the agents, just as we do. Their model, however, involves maximizing utility instead of minimizing costs (which changes the equilibrium structure and all approximation factors like PoS and cost of stabilization entirely), and does not include any facility costs.

\section{Model and Preliminaries}
\label{sec-model}
We are given a set of $m$ facilities $\mathcal{F} = \{f_1, f_2, \dots, f_m\}$ and a set of $n$ agents (which we will also call ``players") $\mathcal{A} = \{1, 2, \dots, n\}$. An agent $i$ can use any facility $f_k$ by paying a connection cost $\cc(i, f_k)$. A pair of agents $(i, j)$ can form connections through facility $f_k$ if they are both using $f_k$. However, if $i$ and $j$ do not use the same facility, then both of them are charged a disconnection cost $\dc(i, j)$. A facility $f_k$ is open if and only if there exists an agent using it. There is a fixed facility cost of $c(f_k) \ge 0$ for any open facility $f_k$.

In much of this paper, we assume each agent uses at most one facility, so the strategy set of an agent consists of $\mathcal{F}$ together with the empty set. A facility assignment $s = \{s_1, s_2, \dots, s_n\}$ denotes the facilities that each agent uses: $s_i$ denotes the facility that agent $i$ uses in assignment $s$. In the case that agent $i$ does not use any facility, let $s_i = \emptyset$ and $\cc(i, s_i) = 0$. A pricing strategy $\gamma = \{\gamma_1, \gamma_2, \dots, \gamma_n\}$ assigns the price for using each facility $f_k$ to every agent $i$. $\gamma_i(f_k)$ is a non-negative number that denotes agent $i$'s share of the facility cost for using $f_k$. $(s, \gamma)$ is a state with assignment $s$ and pricing strategy $\gamma$. Note that agent $i$ only pays its share of the facility cost to $f_k$ if $i$ uses $f_k$, i.e., $\gamma_i(f_k) > 0$ only if $s_i = f_k$. \\

To summarize, the total cost of agent $i$ in a state $(s, \gamma)$ is the sum of the following three parts:

\begin{enumerate}
	\item If $i$ uses facility $f_k$, then there is a connection cost $\cc(i, f_k)$ to $i$.
	\item For each agent $j$ that do not use $s_i$, i.e., $s_i \ne s_j$, there is a disconnection cost $\dc(i, j)$ to both $i$ and $j$. A special case is when both $i$ and $j$ are not using any facility: although $s_i = \emptyset$ and $s_j = \emptyset$, we still say that $s_i \ne s_j$ in this case to make it consistent that the disconnection cost $\dc(i, j)$ is charged to both $i$ and $j$ if $s_i \ne s_j$.
	\item If $i$ uses facility $f_k$, then there is a facility cost $\gamma_i(f_k)$ to $i$.
\end{enumerate}

We denote the total cost of agent $i$ as $\rc_i(s, \gamma)$. Summing up the three types of cost mentioned above:

$$\rc_i(s, \gamma) = \cc(i, s_i) + \sum_{j | s_i \ne s_j} \dc(i, j) + \gamma_i(s_i)$$

For convenience, we denote the cost of agent $i$ without facility cost as $\tc_i(s)$:

$$\tc_i(s) = \cc(i, s_i) + \sum_{j | s_i \ne s_j} \dc(i, j)$$

In this paper, we are interested in the social cost of stable states. The total social cost of a state $(s, \gamma)$ equals the sum of $\rc_i(s, \gamma)$, plus the total cost of all open facilities. For each facility $f_k$, the cost is $c(f_k)$ minus the sum of $\gamma_i(f_k)$ of each agent using $f_k$, i.e., $c(f_k) - \sum_{i|s_i = f_k} \gamma_i(f_k)$. In other words, one can think of each facility as an agent with cost $c(f_k)$, and with other agents paying it the prices $\gamma_i(f_k)$ for using it. The sum of $\gamma_i(f_k)$ cancels out, and the total social cost is actually the sum of $\tc_i(s)$ plus the sum of $c(f_k)$ of open facilities:

\begin{align*}
c(s) &= \sum_{f_k \in \mathcal{F}, f_k \text{ is open}} c(f_k) +
				\sum_{i \in \mathcal{A}} (\cc(i, s_i) + \sum_{j | s_i \ne s_j} \dc(i, j))\\
		 &= \sum_{f_k \in \mathcal{F}, f_k \text{ is open}} c(f_k)
		    + \sum_{i \in \mathcal{A}} \cc(i, s_i) + \sum_{i \in \mathcal{A}} \sum_{j | s_i \ne s_j} \dc(i, j)\\
		 &= \sum_{f_k \in \mathcal{F}, f_k \text{ is open}} c(f_k)
		    + \sum_{i \in \mathcal{A}} \cc(i, s_i) + 2 \sum_{(i, j) | s_i \ne s_j} \dc(i, j)
\end{align*}

We consider $(i, j)$ as an unordered pair, therefore in $\sum_{i \in \mathcal{A}} \sum_{j | s_i \ne s_j} \dc(i, j)$, each unordered pair $(i, j)$ that $s_i \ne s_j$ is counted twice.

In this paper, we study the game in which each agent's goal is to minimize its total social cost, and the central coordinator's goal is to find a budget balanced and stable state $(s, \gamma)$ that (approximately) minimizes the total social cost. A state is budget balanced if each facility $f_k$ is fully paid with the facility cost $c(f_k)$, formally:

\begin{definition}
	A state $(s, \gamma)$ is \textbf{budget balanced} if for each facility $f_k$, $\sum_{i | s_i = f_k} \gamma_i(f_k) = c(f_k)$.
\end{definition}

Before defining the stability of a state, we first define an agent's best response. Consider an agent $i$ with current strategy $s_i=f_k$, and price $\gamma_i(f_k)$ for using this facility. The agent may consider switching to a different facility $f_\ell$, but to correctly evaluate their cost after this switch, the agent needs to know exactly how much they will pay after such a switch. We assume that the agents know their connection costs $\cc(i,f_\ell)$ and their disconnection costs from other agents, as well as which agents are using each facility. What price $\gamma_i(f_\ell)$, however, should they anticipate after switching to their new facility? If the prices depend on the set of agents (or the number of agents) at the facility, then the price might change from the current one being offered. But how reasonable is it for agents to know the exact details of the pricing schemes used by the facilities (which are modeling IXP's or other private enterprises which do not want to reveal their pricing structures)?

To address these issues, in this paper, every agent assumes it will be charged 0 facility cost for joining a new facility. This allows us to not worry about what an agent may know and what price they may anticipate after switching a facility. As the same time, this assumption does not limit our results on stable solutions. This is because no matter what price $\gamma_i(f_\ell)$ an agent may anticipate after switching to facility $f_\ell$, anticipating a price of 0 instead will make it only {\em more} likely to switch. Thus, no matter what the agents' beliefs for prices after switching make sense for a particular setting, {\em a stable solution in our model will still be stable no matter what beliefs about prices $\gamma_i(f_\ell)$ the agents hold, or what price they will actually be charged after switching.} {\bf Thus our results about stable solutions are {\em stronger}: they state that even if the agents are extremely optimistic and believe they can switch to any facility without paying facility cost, then there {\em still} exist good stable solutions.} If they assumed costs higher than 0, then the set of stable solutions would only increase. In other words, if an agent is stable when assuming it will be charged 0 for joining other facilities, then it would also be stable with a higher cost as well.

\begin{definition}
	Given a state $(s, \gamma)$, $s_i'$ is agent $i$'s \textbf{best response} if $\forall s_i'' \ne s_i'$, $\tc_i(s_i', s_{-i}) + \hat{\gamma_i}(s_i') \le \tc_i(s_i'', s_{-i})$, where $\hat{\gamma_i}(s_i') = \gamma_i(s_i)$ if $s_i' = s_i$, and $\hat{\gamma_i}(s_i') = 0$ otherwise. We denote $i$'s best response at state $(s, \gamma)$ as $\BR_i(s, \gamma)$.
\end{definition}

In the definition above, $\hat{\gamma_i}$ is the pricing strategy that agent $i$ assumes would happen after its deviation. If agent $i$ stays at its current facility, then its share of the facility cost does not change. But if $i$ leaves its current facility and joins another one, then it believes that it will be charged 0 facility cost for joining the new facility.

\begin{definition}
	Agent $i$ is stable at state $(s, \gamma)$ if for any strategy $s_i' \ne s_i$ :
	$$\rc_i(s, \gamma) \le \tc_i(s_i', s_{-i}).$$
\end{definition}

In other words, agent $i$ is stable at $(s, \gamma)$ if $s_i$ is $i$'s best response at $(s, \gamma)$.

We define a state $(s, \gamma)$ to be stable if it is budget balanced, and every agent is stable. Intuitively, if a state is not budget balanced, then a facility would not cover its operating cost $c(f_k)$, and thus would not choose to remain open.
\begin{definition}
	A state $(s, \gamma)$ is \textbf{stable} if it is budget balanced, and for each agent $i$, for any strategy $s_i' \ne s_i$ :
	$$\rc_i(s, \gamma) \le \tc_i(s_i', s_{-i}).$$
\end{definition}



Denote an assignment with the minimum total social cost as $s^*$. Our goal is to find stable states $(s, \gamma)$ to approximate the minimum total social cost. We use Price of Stability (PoS) to quantify the quality of a stable state. Given an instance, suppose ($\hat{s}, \gamma$) is the stable state with the smallest total social cost, then PoS is the worst case ratio between $\rc(\hat{s})$ and $\rc(s^*)$ for any instance. A related concept, Price of Anarchy (PoA) is defined as: suppose $(\hat{s}, \gamma)$ is the stable state with the largest total social cost, then PoA is the worst case ratio between $\rc(\hat{s})$ and $\rc(s^*)$ for any instance. So PoS shows the quality of the best stable state, while PoA shows the quality of the worst stable state.

\subsection{Pricing Strategies and Stability}
\label{sec-pricing-stable}
Recall a state $(s, \gamma)$ is stable if it is budget balanced, and every agent $i$ is stable. Suppose there is no other constraint on the pricing strategy, then we ask the following question in order to find a budget balanced state: in an assignment $s$, how much facility cost can we charge an agent while keeping it stable?  To answer this question, we first define a special type of best response: with an assignment $s$, let $\nBR_i(s)$ denote $i$'s best response, given $i$ is forced to stop using $s_i$. In other words, it is the strategy $s_i'\neq s_i$ with the smallest $\tc_i(s_i',s_{-i})$.  Note that if $\BR_i(s,\gamma)\neq s_i$, i.e., if $i$ wants to switch from the state $(s,\gamma)$, then $\BR_i(s,\gamma)=\nBR_i(s)$. But in the case when $i$'s best response is to stay at its current strategy, $\nBR_i(s)$ would denote the ``next best choice'' if $i$ is forced to stop using its current facility.
Intuitively, the ``value'' of facility $s_i$ to agent $i$ is how much $i$'s cost would increase if $i$ is forced to leave $s_i$ and join the next best choice $\nBR_i(s)$. If there are multiple strategies that all satisfy the definition of $\nBR_i(s)$, then we choose an arbitrary one except in one case: we never choose a facility that is closed in $s$ as $\nBR_i(s)$. We can always do this because if there exists such strategy $s_i'$, such that $s_i'$ is a closed facility in $s$, then compare $\tc_i(\emptyset, s_{-i})$ with $\tc_i(s_i', s_{-i})$. The connection cost in $\tc_i(\emptyset, s_{-i})$ is $0$, and the disconnection cost is the same as in $\tc_i(s_i', s_{-i})$, because $i$ would be the only agent using $s_i'$. So it must be $\tc_i(\emptyset, s_{-i}) \le \tc_i(s_i', s_{-i})$, and we define $\nBR_i(s) = \emptyset$ in this case. For every agent $i$ in assignment $s$, define $\Q_i(s) = \tc_i(\nBR_i(s), s_{-i}) - \tc_i(s)$; it is not hard to see that agents would be willing to pay this price in order to use facility $s_i$.

Note that some agents might be unstable even with 0 facility cost, so we also consider the case that agents need to receive payments to be stable. Let $\Delta_i$ denote a payment that agent $i$ receives if it does not deviate at state $(s, \gamma)$, and denote the total payments as $\Delta = \sum_i \Delta_i$. In this paper, the default setting is that agents do not receive payments ($\Delta_i = 0$), but we do consider the cases that agents are allowed to be paid by a central coordinator in Section~\ref{sec-pay-agents} or paid by their neighbors in Section~\ref{sec-paid-peering}.  Then we define the stability with payments as follows: a state $(s, \gamma)$ with payments $\Delta$ is stable if it is budget balanced, and for each agent $i$, for any strategy $s_i'\neq s_i$ :

$$\rc_i(s, \gamma) - \Delta_i \le \tc_i(s_i', s_{-i})$$

 We now show that agent $i$ is stable if $\gamma_i(s_i) - \Delta_i \le \Q_i(s)$.

\begin{lemma}
	\label{lemma-Q}
Given any assignment $s$, pricing strategy $\gamma$, and payments to agents $\Delta$, agent $i$ is stable if $\gamma_i(s_i) - \Delta_i \le \Q_i(s)$.
\end{lemma}

\begin{proof}

	With $\gamma_i(s_i) - \Delta_i \le \Q_i(s)$, by the definition of agent $i$'s total cost and $\Q_i(s)$:

	\begin{align*}
		\rc_i(s, \gamma) - \Delta_i &= \tc_i(s) + \gamma_i(s) - \Delta_i\\
		&\le \tc_i(s) + \Q_i(s) \\
		&= \tc_i(s) + \tc_i(\nBR_i(s), s_{-i}) - \tc_i(s) \\
		&=  \tc_i(\nBR_i(s), s_{-i})
		\end{align*}

	By the definition of $\nBR_i(s)$, for any $s_i' \ne s_i$:
	$$\tc_i(\nBR_i(s), s_{-i}) \le \tc_i(s_i', s_{-i})$$

	Thus, $i$ is stable.
\end{proof}

\section{Single Facility per Agent: Price of Stability}
\label{sec-single}
In the first part of this paper, we show our results in the setting that each agent uses at most one facility.

\subsection{Facility cost $c(f_k) = 0$ for every $f_k$}
In this section, we provide simple baseline results for the case that there is no facility cost. Set the pricing strategy to be $\gamma_i(f_k) = 0$ for any agent $i$ and facility $f_k$, so all solutions are budget balanced. In this special case, for any agent $i$ in assignment $s$, we have $\rc_i(s, \gamma) = \tc_i(s)$. A state $s$ is stable if for each agent $i$ and strategy $s_i'$, $\tc_i(s) \le \tc_i(s_i', s_{-i})$.

Define potential function $\tPhi(s)$ as:

$$\tPhi(s) = \sum_{i \in \mathcal{A}} \cc(i, s_i) + \sum_{(i,j) | s_i \ne s_j} \dc(i, j)$$

When an agent $i$ switches its strategy from $s_i$ to $s_i'$, it is easy to see that the change of $i$'s cost is captured exactly by the change of $\tPhi(s)$, so $\tPhi(s)$ is an exact potential function.

\begin{theorem}
\label{thm-14-single}
	If $\forall k, c(f_k) = 0$, then price of stability is at most 2 and this bound is tight.
\end{theorem}

\begin{proof}

	With the above definition of $\tPhi(s)$, when an agent $i$ switches its strategy from $s_i$ to $s_i'$, it is easy to see that the change of $i$'s cost is captured exactly by the change of $\tPhi(s)$:

	$$\tc_i(s_i', s_{-i}) - \tc_i(s) = \tPhi(s_i', s_{-i}) - \tPhi(s)$$

	Thus, $\tPhi(s)$ is an exact potential function.

	The total social cost in this case is:
	$$c(s) = \sum_{i \in \mathcal{A}} \cc(i, s_i) + 2 \sum_{(i,j) | s_i \ne s_j} \dc(i, j)$$

	Consider the assignment $\hat{s}$ that minimizes $\tPhi$. $\hat{s}$ must be stable, because the exact potential function $\tPhi$ is minimized, so no agent could deviate to lower its cost. Bound the social cost of $\hat{s}$:

	\begin{align*}
	\rc(\hat{s}) &= \sum_{i \in \mathcal{A}} \cc(i, \hat{s}_i) + 2 \sum_{(i,j) | \hat{s}_i \ne \hat{s}_j} \dc(i, j)\\
	&\le 2 \tPhi(\hat{s})\\
	&< 2 \tPhi(s^*)\\
	&\le 2 \rc(s^*)
	\end{align*}

	So PoS is at most 2.

	Consider an example with one facility and two agents. $\cc(1, f_1) = 0$, $\cc(2, f_1) = 1+\epsilon$, $\dc(1,2) = 1$. The only stable state is agent 1 and 2 both do not use $f_k$, i.e., $s_i \ne s_j$, since agent 2 will always want to switch to that state. When $\epsilon$ approaches 0, PoS approaches 2, since the cost of the stable state is 2 (both agents have a disconnection cost of 1), and the cost of optimum is $1+\epsilon$.
\end{proof}

\begin{theorem}
\label{thm-14-poa-single}
  The price of anarchy is unbounded in our setting.
\end{theorem}

\begin{proof}
	 Consider an example with one facility and two agents. $\cc(1, f_1) = 0$, $\cc(2, f_1) = 0$, $\dc(1, 2) = 1$. In an assignment $s$ that agent 1 and 2 do not connect to any facility, $c(s) = 2 \dc(1,2) = 2$, while the optimal solution has $c(s^*) = 0$. This is a stable state since no single agent can reduce their cost without the other agent connection to the facility as well.
\end{proof}

\subsection{Price of Stability for arbitrary facility costs $c(f_k)$}
\label{sec-pos-single}
In this section, we consider the case that for each $f_k$, the facility cost $c(f_k)$ is a fixed constant when $f_k$ is open, regardless of how many agents/connections are using $f_k$. We suppose there is a central coordinator to determine the pricing strategy $\gamma$ that is budget balanced, with no other constraint on $\gamma$.

Note that $\tPhi(s)$ is not a potential function in this setting any more, because agent $i$ also considers the facility cost $\gamma_i(s_i)$ when it deviates to decrease $\rc_i(s)$. Thus, $\tPhi(s)$ does not always decrease when $i$ deviates. We define another potential function $\Phi(s)$:

\begin{align*}
	\Phi(s) &= \sum_{f_k | f_k \text{ is open}} c(f_k) + \sum_{i \in \mathcal{A}} \cc(i, s_i) + \sum_{(i,j) | s_i \ne s_j} \dc(i, j)\\
	&=  \sum_{f_k | f_k \text{ is open}} c(f_k) + \tPhi(s)
\end{align*}

We cannot use the potential method \cite{tardos2007network} to analyze the price of stability in our game directly; in fact our game is not a potential game. For this new potential function $\Phi(s)$, a player could still deviate to lower its cost, while the potential increases. The following lemma, however, shows that when a player deviates and decreases its cost $\tc_i(s)$ (but not necessarily decreases cost $c_i(s,\gamma)$), then $\Phi(s)$ does in fact decrease.


\begin{lemma}
	\label{lemma-phi-basic}
	In an assignment $s$, if any agent $i$ switches its strategy to $s_i'$ such that $\tc_i(s_i', s_{-i}) < \tc_i(s)$ and $s_i'$ does not contain any closed facility in $s$, then $\Phi(s_i', s_{-i}) < \Phi(s)$.
\end{lemma}

\begin{proof}
	First, $\tPhi(s)$ is an exact potential function when there is no facility cost, so $\tPhi(s_i', s_{-i}) < \tPhi(s)$ when $\tc_i(s_i', s_{-i}) < \tc_i(s)$. Also, because $s_i'$ does not contain any closed facility in $s$, then there is no new facility open in $(s_i', s_{-i})$ compared to $s$, and the total facility cost $\sum_{f_k | f_k \text{ is open}} c(f_k)$ is non-increasing compared to that in $s$. Thus, $\Phi(s_i', s_{-i}) < \Phi(s)$.
\end{proof}

Now we use the above potential to prove bounds on the price of stability. While a single player changing its strategy to decrease its cost might actually increase the value of the potential $\Phi(s)$, we give a series of coalitional deviations (i.e., groups of players switching strategies simultaneously) so that the potential is guaranteed to decrease after each such deviation, and so that the cost of the resulting stable solution is not too large.

\begin{theorem}
	\label{thm-134-fix3-central-single}
    The price of stability is at most 2, and this bound is tight. In other words, there exists a stable state $(s,\gamma)$ with cost at most twice that of optimum.
\end{theorem}

\begin{proof}
	We define a coalitional deviation process that converges to a stable state, with $\Phi(s)$ decreasing in each step of the process.

	Start with the optimal assignment $s^*$. If there exists an agent $i$ such that when $i$ switches to a strategy $s_i'$, in which $s_i'$ does not contain any closed facility in $s^*$, then $\tc_i(s_i', s^*_{-i}) < \tc_i(s^*)$, then let agent $i$ switch to $s_i'$. Select another agent to repeat this process until no such agent exists. In other words, each agent is now stable if they assume they are not charged any facility cost. By Lemma~\ref{lemma-phi-basic}, $\Phi(s)$ decreases during each step in this process. Let $s$ be the current state.

	We know that no agent $i$ can decrease $\tc_i(s)$ by switching to another strategy that does not contain any closed facility in $s$. Note that even if $i$ switches to a closed facility in $s$, it would not be able to lower $\tc_i(s)$. This is because if $i$ switches to a closed facility $f_k$, $i$ will pay a connection cost of $\cc(i, f_k)$, and because $i$ is the only agent using $f_k$, so $i$'s disconnection cost would not decrease. Thus, in this ``stable'' state, every agent is stable if it is charged 0 facility cost. By Lemma~\ref{lemma-Q}, for any agent $i$ that uses $f_k$, we can charge $\Q_i(s)$ to agent $i$ while keeping it stable:

	\begin{align}
	 \Q_i(s) &= \tc_i(\nBR_i(s), s_{-i}) - \tc_i(s) \nonumber \\
		&= \cc(i, \nBR_i(s)) + \sum_{j | s_j = f_k} \dc(i,  j) - \cc(i, f_k) - \sum_{j | s_j = \nBR_i(s)} \dc(i,  j) \label{thm-134-fix3-central-single-eq1}
	\end{align}

	For each facility $f_k$, consider the following two cases:\\

	\textbf{Case 1, $c(f_k) > \sum_{i|s_i = f_k} \Q_i(s)$.} $c(f_k)$ is greater than the total payments we can charge all agents using $f_k$ while keeping them stable. In this case, we close $f_k$ and let each agent $i$ using $f_k$ in $s$ switch its strategy $\nBR_i(s)$. Denote the assignment after closing $f_k$ as $s'$, then consider the value change of $\Phi(s)$:

	\begin{align*}
	\Phi(s') - \Phi(s)
	&= -c(f_k) + \sum_{i|s_i = f_k} (\cc(i, \nBR_i(s)) - \cc(i, f_k)) \\
	&+ \sum_{(i,j) | s_i = s_j = f_k, s'_i \ne s'_j} \dc(i, j) - \sum_{(i,j) | s_i \ne s_j, s'_i = s'_j} \dc(i, j)\\
	\end{align*}

   Note that only agents using $f_k$ in $s$ change their strategies, and all other agents keep their strategies at $s$. Thus, the newly disconnected agent pairs in $s'$ are at most all the connected pairs using $f_k$ in $s$:

	\begin{align*}
	\sum_{(i,j) | s_i = s_j = f_k, s'_i \ne s'_j} \dc(i, j)
	&\le \sum_{(i, j) |  s_i = s_j = f_k } \dc(i, j) \\
	&= \frac{1}{2} \sum_{i| s_i = f_k} \sum_{j | s_j = f_k} \dc(i,  j) \\
	&\le \sum_{i|s_i = f_k} \sum_{j | s_j = f_k} \dc(i,  j)
  \end{align*}

	Also, the newly connected agent pairs in $s'$ are at least those created by the agents in $f_k$ deviating to their $\nBR_i(s)$, i.e.,

	$$ \sum_{(i,j) | s_i \ne s_j, s'_i = s'_j} \dc(i, j) \ge  \sum_{i| s_i = f_k} \sum_{j | s_j = \nBR_i(s)} \dc(i,  j) $$

  With the condition of \textbf{Case 1} and Equation \ref{thm-134-fix3-central-single-eq1}, we can bound $\Phi(s') - \Phi(s)$ by:

	\begin{align*}
	&\Phi(s') - \Phi(s)\\
	&= -c(f_k) + \sum_{i|s_i = f_k} (\cc(i, \nBR_i(s) - \cc(i, f_k))
	+ \sum_{(i,j) | s_i = s_j = f_k, s'_i \ne s'_j} \dc(i, j) - \sum_{(i,j) | s_i \ne s_j, s'_i = s'_j} \dc(i, j)\\
	&\le -c(f_k) + \sum_{i|s_i = f_k} (\cc(i, \nBR_i(s) - \cc(i, f_k))
	+ \sum_{i|s_i = f_k} \sum_{j | s_j = f_k} \dc(i,  j) - \sum_{i|s_i = f_k} \sum_{j | s_j = \nBR_i(s)} \dc(i,  j)\\
	&= -c(f_k) + \sum_{i|s_i = f_k} (\cc(i, \nBR_i(s)) + \sum_{j | s_j = f_k} \dc(i,  j)
	- \cc(i, f_k) - \sum_{j | s_j = \nBR_i(s)} \dc(i,  j))\\
	&= -c(f_k) + \sum_{i|s_i = f_k} (\tc_i(\nBR_i(s), s_{-i}) - \tc_i(s))\\
	&= -c(f_k) + \sum_{i|s_i = f_k} \Q_i(s)\\
	&< 0 \\
  \end{align*}

	Thus, $\Phi(s)$ decreases after $f_k$ is closed and every agent $i$ using it switches its strategy to $\nBR_i(s)$. \\

	Then repeat the above two steps: let agents switch strategies to reach a ``stable'' state $s$ that there is no agent $i$ could switch its strategy to $s_i'$, in which $s_i'$ does not contain any closed facility in $s$ and $\tc_i(s_i', s_{-i}) < \tc_i(s)$. Then if there exist a facility $f_k$ that satisfies the condition in \textbf{Case 1}, we close $f_k$ and let every agent $i$ using it switch its strategy to $\nBR_i(s)$. We repeat these two steps until we reach state $s$ that every agent is ``stable'', and every open facility does not satisfy \textbf{Case 1}. Note that $\Phi(s)$ decreases in each step, so this process always converges to such an assignment $s$. Then each open facility must satisfy the following \textbf{Case 2}:

	\textbf{Case 2, $c(f_k) \le \sum_{i|s_i = f_k} Q_i(s)$.} In this case, we set $\gamma_i(f_k) = \Q_i(s)$ for each agent $i$. By Lemma~\ref{lemma-Q}, every agent is stable with $\gamma_i(f_k) = \Q_i(s)$. Also, by definition, $\Q_i(s) = \tc_i(\nBR_i(s), s_{-i}) - \tc_i(s)$. Because every agent is stable without considering facility costs, so $\tc_i(s) \le \tc_i(\nBR_i(s), s_{-i})$, then we are charging a non-negative facility cost to each agent. By the condition of \textbf{Case 2}, $\sum_{i|s_i = f_k} \gamma_i(f_k) \ge c(f_k)$. If $\sum_{i|s_i = f_k} \gamma_i(f_k) > c(f_k)$, then to get a budget-balanced cost assignment, we can lower the facility cost of some agents to make the sum of $\gamma_i(f_k)$ exactly $c(f_k)$, because the agents would not deviate with $\gamma_i(f_k)$, then they would not deviate will a lower facility cost. Thus, we have reached a stable state.

   $\Phi(s)$ decreases in each deviation, so it is an ordinal potential function for our deviation processes.

	The total social cost is:

	$$c(s) = \sum_{f_k | f_k \text{ is open}} c(f_k) + \sum_{i \in \mathcal{A}} \cc(i, s_i) + 2\sum_{(i, j)| s_i \ne s_j} \dc(i, j)$$

	Denote the final stable state as $\hat{s}$. Similar to the analysis in Theorem \ref{thm-14-single},

	$$\rc(\hat{s}) \le 2 \Phi(\hat{s}) <  2 \Phi(s^*) \le 2 \rc(s^*)$$

	Thus, PoS is at most 2.\\

	See Theorem~\ref{thm-14-single} for the lower bound example showing this bound is tight (even if $c(f_k) = 0$).
\end{proof}


The above price of stability result can be easily generalized to approximately stable solutions as well. We say a state $(s, \gamma)$ is $\alpha$-approximate stable if it is budget balanced, and no agent could deviate to lower its cost to $\frac{1}{\alpha}$ of its current cost:

\begin{definition}
	A state $(s, \gamma)$ is \textbf{$\alpha$-approximate stable} if it is budget balanced, and for each agent $i$, for any strategy $s_i' \ne s_i$ :
	$$\rc_i(s, \gamma) \le \alpha\cdot \tc_i(s_i', s_{-i})$$
\end{definition}

\begin{theorem}
	\label{thm-alpha-single}
	 There always exists an $\alpha$-approximate stable state $(\hat{s}, \gamma)$ such that $\frac{c(\hat{s})}{c(s^*)} \le \frac{2}{\alpha}$.
\end{theorem}

\begin{proof}
	For any assignment $s$, consider the following potential function:
	$$\Phi_{\alpha}(s) = \sum_{i \in \mathcal{A}} \cc(i, s_i) + \alpha \sum_{(i,j)| s_i \ne s_j} \dc(i, j) + \sum_{f_k \in \mathcal{F} | f_k \text{ is open}}$$

   The proof is essentially the same as that of Theorem~\ref{thm-134-fix3-central-single}, except agent $i$ only would like to deviate from assignment $s$ to $s'$ if $c_i(s') < \frac{1}{\alpha} c_i(s)$, and the potential used is the one above.

  Remember the total social cost is:

	$$\rc(s) = \sum_{i \in \mathcal{A}} \cc(i, s_i) + 2 \sum_{(i,j)| s_i \ne s_j} \dc(i, j) + \sum_{f_k \in \mathcal{F} | f_k \text{ is open}}$$

	Thus, if we start from $s^*$, proceed with our coalitional deviation process to form an $\alpha$-approximate Equilibrium $(\hat{s}, \gamma)$, and make sure $\Phi_{\alpha}(s)$ is non-increasing in each step, then $\frac{\rc(\hat{s})}{\rc(s^*)} \le \frac{2}{\alpha}$.
\end{proof}

This theorem implies that, in particular, the optimum solution $s^*$ is a 2-approximate stable state, i.e., no player can improve their cost by more than a factor of 2 by switching its facility.

\section{Payments to Form Good Stable Solutions}

\subsection{Agents paying each other: ``Paid peering"}
\label{sec-paid-peering}
In this section, we consider the case that agents can pay each other to stabilize the optimal assignment. Formally, for a pair $(i, j)$ such that $s_i = s_j = f_k$, $i$ can pay $j$ up to $\dc(i, j)$ in order to discourage $j$ from leaving facility $f_k$ and thus disconnecting from $i$. Given the asymmetry of agent costs (due to connection costs), it may make sense for agents to give their ``friends" extra incentives to connect with them using a particular facility. Of course, agent $i$ would never voluntarily pay agent $j$ more than $j$'s value to $i$, i.e., more than $\dc(i, j)$. Such payments make sense in general settings of group formation, and make sense in our motivating IXP setting as well: when two ISP's decide to make a peering arrangement to exchange traffic after joining a common IXP, it is often the case that they make a {\em paid} peering contract \cite{shrimali2006paid}, in which one ISP pays the other for the traffic exchange, thus giving it extra incentive to remain connected to their joint facility.

Let $p_{ij}$ denote the payments that agent $i$ pays its neighbor $j$ to discourage it from leaving the facility they share. $p_{ij} \ge 0$ means $i$ pays $j$, and $p_{ij} < 0$ means $i$ receives payment from $j$. For any pair of agents $(i, j)$, we have $p_{ji} = - p_{ij}$. In this section, $\Delta_i$ denotes the total payments that agent $i$ receives from its neighbors minus the total payments $i$ pays its neighbors. In other words, $\Delta_i = \sum_{j|s_i = s_j} p_{ji}$. We abuse the notation to allow $\Delta_i$ to be negative, in which case $i$ pays more than receives from its neighbors. We consider stability with payments defined in Section~\ref{sec-pricing-stable}. It is easy to see that Lemma~\ref{lemma-Q} still holds with this modified definition of $\Delta_i$.

In the optimal assignment $s^*$ with a pricing strategy $\gamma$, consider the stability of every agent using $f_k$: by Lemma~\ref{lemma-Q}, we know every agent $i$ would be stable if $\gamma_i(s^*_i) - \Delta_i \le \Q_i(s^*)$. For a pair of agents $(i, j)$ using $f_k$ in $s^*$, suppose $\Q_i(s^*) \ge 0$, and $\Q_j(s^*) < 0$, which means we can get some payments from $i$ (to pay the facility or its neighbors) while keeping it stable, but $j$ needs to be paid to become stable at $s^*$. Thus, it makes sense for $i$ to pay $j$ to stop it from deviating, but $i$ would not pay more than $\dc(i, j)$, which is the maximum increase of $i$'s cost as a result of $j$'s deviation.

\begin{theorem}
	\label{paid-peering-thm1-single}
	 If we allow agents to pay their neighbors, and $i$ pays $j$ no more than $\dc(i, j)$, then there exist $\gamma$ and payments of players to each other so that the resulting solution $(s^*,\gamma)$ is stable, with $s^*$ being the solution minimizing social cost. In other words, the price of stability becomes 1.
\end{theorem}

\begin{proof}
We construct a circulation network \cite{kleinberg2006algorithm} for each facility $f_k$ as follows: start from the optimal assignment $s^*$. Create a node for each agent $i$ such that $s^*_i = f_k$, and we set a supply of $\Q_i(s^*)$ to it (note that this value might be negative, in which case the node has a demand instead of supply). For each pair of agents $(i, j)$, we create directed edges from $i$ to $j$ and from $j$ to $i$, both with capacity $\dc(i, j)$. Create a node $f_k$ with a supply of $-c(f_k)$ and an edge from each node to $f_k$ with infinite capacity. Finally, create a dummy node $z$ with a demand of the sum of supplies of all other nodes, and add an edge from $f_k$ to $z$ with infinite capacity. This is to make sure the total supply equals the total demand in the network.

Suppose there is a feasible solution, i.e., a valid circulation which satisfies all the supplies and demands, and obeys the capacities of the edges. Then we can use the flow on each edge to create a stable state, as follows. First, denote the flow from any node $i$ to $j$ as $v_{ij}$. For each pair of nodes $i$ and $j$ such that $s^*_i = s^*_j = f_k$, set $p_{ij} = v_{ij} - v_{ji}$ and $p_{ji} = v_{ji} - v_{ij}$. Also, for every agent $i$ such that $s^*_i = f_k$, set $\gamma_i(f_k) = v_{if_k}$. A feasible solution guarantees that facility $f_k$ is fully paid for, because $\sum_{i|s^*_i = f_k} \gamma_i(f_k) = \sum_{i|s^*_i = f_k} v_{if_k} \geq c(f_k)$. Also, every agent is stable. To see this, first, by the definition of $\Delta_i$ in this section, $\Delta_i = \sum_{j|s^*_i = s^*_j} p_{ji} = \sum_{j|s^*_i = s^*_j} (v_{ji} - v_{ij})$. For agent $i$ such that $s^*_i = f_k$, the supply of node $i$ is $\Q_i(s^*)$, which equals the total flow going out of $i$ minus the total flow going into $i$:

$$\Q_i(s^*) = v_{if_k} + \sum_{j|s^*_i = s^*_j = f_k} (v_{ij} - v_{ji}) = \gamma_i(f_k) - \Delta_i.$$

By Lemma~\ref{lemma-Q}, every agent $i$ is stable if $\gamma_i(s^*_i) - \Delta_i = \Q_i(s^*)$, so every agent is stable.

To prove the theorem, all we need to show is that this circulation network is feasible.
By a standard Max-Flow and Min-Cut analysis \cite{kleinberg2006algorithm}, if for every subset of nodes in the circulation network, the total supply of the subset plus the total capacity of edges going into the subset is non-negative, then the circulation network is feasible. We now argue that this is true.

First consider any subset that includes $z$. If the subset does not include $f_k$, then there must be an edge with infinite capacity going into the subset, so the conclusion holds.

Next, consider a subset that includes $f_k$ and $z$. If the subset does not include all agents in $f_k$, then there must be an edge with infinite capacity going into the subset, so the conclusion holds. If the subset does include all agents in $f_k$, then by the definition of $z$, the total supply is 0.

Then, consider a subset that includes $f_k$ but not $z$. If the subset does not include all agents in $f_k$, the conclusion still holds. If the subset is actually all the nodes in the network, then the total supply is:
	$$ \sum_{i \in \mathcal{A}} \Q_i(s^*) - c(f_k),$$
and there are no incoming edges into this set of nodes. Suppose for the sake of forming a contradiction, that the total supply is negative. Then consider an assignment $s'$ such that $f_k$ is closed, and every agent $i$ that uses $f_k$ in $s^*$ switches its strategy to $\nBR_i(s^*)$. For any agent $j$ that does not use $f_k$ in $s^*$, $j$ stays at $s^*_j$. It it easy to see that $\tc_j(s') \le \tc_j(s^*)$ for every $j$ such that $s^*_j \ne f_k$. For any agent $i$ with $s^*_i = f_k$, it must be $\tc_i(s') \le \tc_i(\nBR_i(s^*), s^*_{-i})$. This is because $\tc_i(\nBR_i(s^*), s^*_{-i})$ is the cost if only $i$ switches to $\nBR_i(s^*)$ while all other agents stay at $s^*$, while in $s'$ all agents using $f_k$ switch their strategies. Because agents using $\nBR_i(s^*)$ in $s^*$ all stay at $s^*$, $i$ would not get any ``unexpected cost'' in $s'$, so $\tc_i(s') \le \tc_i(\nBR_i(s^*), s^*_{-i})$. Also, because $f_k$ is closed in $s'$, the facility cost decreases by at least $c(f_k)$. No new facility will open in $s'$ (compared to $s^*$) because we have excluded this possibility in the definition of $\nBR_i(s^*)$. Thus, the total social cost of $s'$ increases by at most:
  $$-c(f_k) + \sum_{i | s^*_i = f_k} (\tc_i(\nBR_i(s^*), s^*_{-i}) - \tc_i(s^*))
  = -c(f_k) + \sum_{i | s^*_i = f_k} \Q_i(s^*).$$
By our assumption, this number is negative, which means $s'$ has less total cost than $s^*$, which contradicts the fact that $s^*$ is optimal. Thus, the total supply must be non-negative.

Finally, consider a subset of nodes that does not include $f_k$. Suppose there exists a subset $\mathcal{B}$, such that the total supply of nodes in $\mathcal{B}$ plus the total capacity of edges going into the subset is negative, i.e.,

	\begin{align}
	\sum_{i \in \mathcal{B}} \Q_i(s^*) + \sum_{i \in \mathcal{B}, j \notin \mathcal{B}, s^*_j = f_k} \dc(i, j) &< 0 \nonumber \\
	\sum_{i \in \mathcal{B}} [\cc(i, \nBR_i(s^*)) + \sum_{s^*_j = f_k, j \ne i} \dc(i, j) - \cc(i, f_k) - \sum_{s^*_j = \nBR_i(s^*)} \dc(i, j)] + \sum_{i \in \mathcal{B}, j \notin \mathcal{B}, s^*_j = f_k} \dc(i, j) &< 0 \nonumber \\
	\sum_{i \in \mathcal{B}} \cc(i, \nBR_i(s^*)) + \sum_{i \in \mathcal{B}} \sum_{s^*_j = f_k, j \ne i} \dc(i, j) - \sum_{i \in \mathcal{B}} \cc(i, f_k) - \sum_{i \in \mathcal{B}} \sum_{s^*_j = \nBR_i(s^*)} \dc(i, j) + \sum_{i \in \mathcal{B}, j \notin \mathcal{B}, s^*_j = f_k} \dc(i, j) &< 0 \label{paid-peering-eq1-single}
	\end{align}

Decompose $\sum_{i \in \mathcal{B}} \sum_{s^*_j = f_k, j \ne i} \dc(i, j)$ into two parts ($j \in \mathcal{B}$ or $j \notin \mathcal{B}$):

\begin{align*}
	\sum_{i \in \mathcal{B}} \sum_{s^*_j = f_k, j \ne i} \dc(i, j)
	&= \sum_{i \in \mathcal{B}} (\sum_{j \in \mathcal{B}, s^*_j = f_k, j \ne i} \dc(i, j) +  \sum_{j \notin \mathcal{B}, s^*_j = f_k} \dc(i, j) ) \\
	&= \sum_{i \in \mathcal{B}} \sum_{j \in \mathcal{B}, j \ne i} \dc(i, j) + \sum_{i \in \mathcal{B}} \sum_{j \notin \mathcal{B}, s^*_j = f_k} \dc(i, j) \\
	&= 2 \sum_{i, j \in \mathcal{B}}  \dc(i, j) + \sum_{i \in \mathcal{B}, j \notin \mathcal{B}, s^*_j = f_k} \dc(i, j) \\
\end{align*}

Remember we only create nodes for every agent $j$ such that $s^*_j = f_k$, and $\mathcal{B}$ is a subset of the nodes, so $\{j \in \mathcal{B}, s^*_j = f_k, j \ne i\} = \{j \in \mathcal{B}, j \ne i\}$ in the first line of the inequality above.\\

Together with Inequality~\ref{paid-peering-eq1-single},
\begin{align*}
	&\sum_{i \in \mathcal{B}} \cc(i, \nBR_i(s^*))
	+ 2 \sum_{i, j \in \mathcal{B}}  \dc(i, j)
	+ 2 \sum_{i \in \mathcal{B}, j \notin \mathcal{B}, s^*_j = f_k} \dc(i, j)
	- \sum_{i \in \mathcal{B}} \cc(i, f_k)
	- \sum_{i \in \mathcal{B}} \sum_{s^*_j = \nBR_i(s^*)} \dc(i, j) < 0
\end{align*}

Which is equivalent to:
\begin{equation}
	\label{paid-peering-eq2-single}
	\sum_{i \in \mathcal{B}} \cc(i, \nBR_i(s^*)) + 2 \sum_{i, j \in \mathcal{B}}  \dc(i, j) + 2\sum_{i \in \mathcal{B}, j \notin \mathcal{B}, s^*_j = f_k} \dc(i, j)
	< \sum_{i \in \mathcal{B}} \cc(i, f_k) + \sum_{i \in \mathcal{B}} \sum_{s^*_j = \nBR_i(s^*)} \dc(i, j)
\end{equation}

Consider an assignment $s'$: start from $s^*$ and let every agent $i$ in $\mathcal{B}$ switch to $\nBR_i(s^*)$. The total facility cost of $s'$ is no more than in $s^*$ because no agent would switch to a closed facility in $s^*$ by the definition of $\nBR_i(s^*)$. The total connection and disconnection cost in $s'$ compared to $s^*$ increases by at most the left hand side of Inequality~\ref{paid-peering-eq2-single}, and decreases by the right hand side of it. Thus the total social cost of $s'$ is less than $s^*$, which is a contradiction, and so such a subset $\mathcal{B}$ does not exist. Since we have now proven that the circulation network is feasible, this completes the proof of the theorem.
\end{proof}

\begin{theorem}
	\label{increase-edge-weight-thm-single}
    Let $s^*$ be a solution with minimum social cost. Then, doubling the disconnection costs makes $s^*$ become a stable solution.
\end{theorem}

\begin{proof}
	Consider the circulation network we constructed in the proof of Theorem \ref{paid-peering-thm1-single}, suppose we are given a feasible circulation, and denote the flow from any node $i$ to $j$ as $v_{ij}$. Between any pair of agents $i$ and $j$, define $f(i, j) = |v_{ij} - v_{ji}|$, which is the absolute value difference of flow between $i$ and $j$. Intuitively, if there exist flows both from $i$ to $j$ and $j$ to $i$, we cancel out the flow in one direction. Next, we increase the disconnection cost between $i$ and $j$ by $f(i, j)$. Denote the new disconnection cost as $d'(i, j)$, i.e., $d'(i, j) = \dc(i, j) + f(i, j)$. Because the capacity of the edge between $i$ and $j$ is $\dc(i,  j)$, we know $d'(i, j) \le 2 \dc(i, j)$.

	We will show that with this new disconnection cost $d'(i, j)$ for every pair of nodes $(i, j)$, $s^*$ is a stable solution. For any node $j$ such that $s^*_j = f_k$ with a positive demand, by the analysis from Theorem \ref{paid-peering-thm1-single}, in a feasible solution of the circulation network, the total flow going into node $j$ equals to the total payments $j$ needs to receive to be stabilized. By increasing the disconnection cost from all $i$ that pays $j$ in the feasible solution by $f(i, j)$, we increase the cost for leaving $f_k$ by the sum of such $f(i, j)$, so $j$ is stabilized. For a node $i$ with a positive supply in the network, the supply equals the total flow going out of $i$, and agent $i$ is stable with a total payment of $\Q_i(s^*)$ to its neighbors and $f_k$. If there is a flow of $f(i, j)$ from $i$ to $j$, then by converting this payment to the increase of edge weight, $i$'s total social cost stays the same, so $i$ is still stable.
\end{proof}

\subsection{Paying agents directly to stabilize $s^*$}
\label{sec-pay-agents}
In this section, we take on the role of a central coordinator, who is paying the agents in order to stabilize the optimum solution $s^*$. We study the relationship between the Price of Stability and the minimum total payments required to stabilize $s^*$. We use the notation of $\Delta_i$ and stability with payments defined in Section~\ref{sec-pricing-stable}. $\Delta_i$ represent the payment each agent $i$ receives from the central coordinator, and the total payments are $\Delta = \sum_i \Delta_i$.

\begin{lemma}
	\label{pay-agents-lemma-delta-single}
	The following pricing strategy $\gamma$ and payments strategy $\Delta$ stabilizes the optimal assignment $s^*$: For each facility $f_k$,

	\textbf{Case 1. $\forall i$ such that $s^*_i = f_k$ and $\Q_i(s^*) < 0$, set $\gamma_i(f_k) = 0$ and $\Delta_i = -\Q_i(s^*)$.}

	\textbf{Case 2. $\forall i$ such that $s^*_i = f_k$ and $\Q_i(s^*) \ge 0$, set $\gamma_i(f_k) = \Q_i(s^*)$ and $\Delta_i = 0$.}\\
\end{lemma}

\begin{proof}

	To prove $(s^*, \gamma)$ is stable and budget balanced, it is enough to show that every agent is stable, and for each facility $f_k$, $\sum_{i|s^*_i = f_k} \gamma_i(f_k) \ge c(f_k)$. If $\sum_{i|s^*_i = f_k} \gamma_i(f_k) > c(f_k)$, then we can always lower some of $\gamma_i(f_k)$ while keeping $i$ stable.

	First, note that in both \textbf{Case 1} and \textbf{Case 2}, $\gamma_i(f_k) - \Delta_i = \Q_i(s^*)$. Thus, by Lemma~\ref{lemma-Q}, every agent $i$ is stable at state $(s^*, \gamma)$ with a payment $\Delta_i$.

	We now prove $(s^*, \gamma)$ is budget balanced. Suppose to the contrary that there exists facility $f_k$ such that $\sum_{i|s^*_i = f_k} \gamma_i(f_k) < c(f_k)$. Then we close $f_k$ and let every agent $i$ using $f_k$ in $s^*$ switch its strategy to $\nBR_i(s^*)$. Denote this new assignment as $\hat{s}$. Because all other agents using $\nBR_i(s^*)$ in $s^*$ stay at $s^*$, we know that $\tc_i(\hat{s})$ is at most $\tc_i(\nBR_i(s^*), s^*_{-i})$, i.e.:

	\begin{equation}
		\label{pay-agents-lemma-delta-single-eq1}
	\tc_i(\hat{s}) \le \tc_i(\nBR_i(s^*), s^*_{-i}).
	\end{equation}

	For every agent $i$ in \textbf{Case 1}, because $\Q_i(s^*) < 0$, by the definition of $\Q_i(s^*)$:

	$$\Q_i(s^*) = \tc_i(\nBR_i(s^*), s^*_{-i}) - \tc_i(s^*) < 0$$

	Combine it with Inequality~\ref{pay-agents-lemma-delta-single-eq1}:

	$$\tc_i(\hat{s}) \le \tc_i(\nBR_i(s^*), s^*_{-i}) < \tc_i(s^*)$$

	Remember $\gamma_i(f_k) = 0$ in this case, so:
	
	$$\tc_i(\hat{s}) \le \tc_i(s^*) =  \tc_i(s^*) + \gamma_i(f_k)$$

	For every agent $i$ in \textbf{Case 2}, $\gamma_i(f_k) = \Q_i(s^*)$, combine with Inequality~\ref{pay-agents-lemma-delta-single-eq1} and the definition of $\Q_i(s^*)$:

	\begin{align*}
		\tc_i(\hat{s}) &\le \tc_i(\nBR_i(s^*), s^*_{-i}) \\
		&= \Q_i(s^*) + \tc_i(s^*)\\
		&= \gamma_i(f_k) + \tc_i(s^*)
	\end{align*}

	Thus, $\tc_i(\hat{s}) \le \tc_i(s^*) +  \gamma_i(f_k)$ in both cases. Sum up all $i$ such that $s^*_i = f_k$:

	$$\sum_{i| s^*_i = f_k} \tc_i(\hat{s}) < \sum_{i| s^*_i = f_k} \tc_i(s^*) + \sum_{i| s^*_i = f_k} \gamma_i(f_k)$$

	Because we assume $\sum_{i|s^*_i = f_k} \gamma_i(f_k) < c(f_k)$, so≈:

	$$\sum_{i| s^*_i = f_k} \tc_i(\hat{s}) < \sum_{i| s^*_i = f_k} \tc_i(s^*) + c(f_k)$$

	Only agents that use $f_k$ in $s^*$ change their strategies in $\hat{s}$, so for every agent $i$ such that $s^*_i \ne f_k$, we have $\tc_i(\hat{s}) \le \tc_i(s^*)$. Now consider the total social cost of state $s$: besides the costs change between $\tc(s^*)$ and $\tc(\hat{s})$, the total facility cost decreases by at least $c(f_k)$. Note that no new facility will open in $\hat{s}$ (compared to $s^*$) by the definition of $\nBR_i(s^*)$. Therefore:

	$$\rc(\hat{s}) - \rc(s^*) \le \sum_i (\tc_i(\hat{s}) - \tc_i(s^*)) - c(f_k) < 0$$

	This contradicts the fact that $s^*$ is the optimal solution. Thus, $(s^*, \gamma)$ must be budget balanced.

\end{proof}

In the following theorem, we show that the total payments $\Delta$ required to stabilize $s^*$ is only a fraction of the social cost of the optimal solution. Actually, there is a tradeoff between $\Delta$ and $PoS$: when $PoS$ is large, e.g., $PoS = 2$, we only need to pay $\frac{1}{5} \rc(s^*)$ to stabilize $s^*$, which is only a small fraction of $\rc(s^*)$. Thus when $PoS$ is small, there already exist good stable solutions by definition of $PoS$, and when $PoS$ is large, only a relatively small amount of payments are necessary to stabilize $s^*$.

\begin{theorem}
	\label{pay-agents-thm-single}
	For any instance, $\frac{\Delta}{\rc(s^*)} \le 1 - \frac{2}{5} PoS$, where $\Delta$ is the payment needed to stabilize $s^*$.
\end{theorem}

Before we prove this theorem, we define some extra notation. Let $b_i$ denote the strategy of agent $i$ such that $\tc_i(b_i, s^*_{-i})$ is minimized. We can always find such a $b_i$ with either $b_i = \emptyset$, or $b_i$ is open in $s^*$. This is because if there exists $b_i'$ that minimizes $\tc_i(b_i', s^*_{-i})$ and $b_i'$ is closed in $s^*$, then it must be $\tc_i(\emptyset, s^*_{-i}) \le \tc_i(b_i', s^*_{-i})$ because $i$ would be the only agent using $b_i'$. Then we define an assignment $b^i$ for each agent $i$: start from $s^*$, only let agent $i$ switch its strategy to $b_i$ and all other agents stay at their facility in $s^*$, i.e., $b^i = (b_i, s^*_{-i})$.

For any subset of agents $\mathcal{A}_1 \subseteq \mathcal{A}$, let $\CC{s}{\mathcal{A}_1}$ denote the total connection cost for agents in $\mathcal{A}_1$ in assignment $s$:

$$\CC{s}{\mathcal{A}_1} = \sum_{i \in \mathcal{A}_1} \cc(i, s_i) $$

With assignment $s$, for any two subsets of agents $\mathcal{A}_1 \subseteq \mathcal{A}$, $\mathcal{A}_2 \subseteq \mathcal{A}$, define $\DC{s}{\mathcal{A}_1}{\mathcal{A}_2}$ as:

$$\DC{s}{\mathcal{A}_1}{\mathcal{A}_2} = \sum_{(i, j)| i \in \mathcal{A}_1, j \in \mathcal{A}_2, s_i\neq s_j} \dc(i, j)$$

$\DC{s}{\mathcal{A}_1}{\mathcal{A}_2}$ represents the disconnection cost between $\mathcal{A}_1$ and $\mathcal{A}_2$ in $s$. Note that $(i, j)$ is still an unordered pair here. $\DC{s}{\mathcal{A}_1}{\mathcal{A}_1}$ represents the disconnection cost inside $\mathcal{A}_1$.

By the definition of $\tPhi(s)$, we can rewrite it using the notation defined above as:

$$\tPhi(s) = \CC{s}{\mathcal{A}} + \DC{s}{\mathcal{A}}{\mathcal{A}}$$

Also, we can rewrite $\tc(s)$ as:

$$\tc(s) =  \CC{s}{\mathcal{A}} + 2 \DC{s}{\mathcal{A}}{\mathcal{A}}$$

The following lemma gives a condition that would directly imply Theorem~\ref{pay-agents-thm-single}. We first show this lemma, and then prove the theorem.

\begin{lemma}
	\label{pay-agents-lemma-s-single}
	For any $\tau \ge 2$, if there exists a state $s$, such that $\sum_{i \in \mathcal{A}} \tc_i(b^i)
	\ge  \frac{2}{\tau} \tPhi(s)$, then $\frac{\Delta}{\rc(s^*)} \le 1 - \frac{1}{\tau} PoS$.
\end{lemma}

\begin{proof}
	Consider the payment strategy $\Delta$ defined in Lemma \ref{pay-agents-lemma-delta-single}. We first prove that $\Delta_i = \tc_i(s^*) - \tc_i(b^i)$. We discuss the two cases of agents in Lemma \ref{pay-agents-lemma-delta-single} separately.

	In \textbf{Case 1}, $\Q_i(s^*) < 0$. By definition, $\Q_i(s^*) = \tc_i(\nBR_i(s^*), s^*_{-i}) - \tc_i(s^*) < 0$. Remember $b_i$ is the strategy to minimize $\tc_i(b_i, s^*_{-i})$. Because $\tc_i(\nBR_i(s^*), s^*_{-i}) < \tc_i(s^*)$, we know that $s^*_i$ does not minimize $\tc_i$ if $i$ is the only one allowed to change its strategy, i.e. $b_i \ne s^*_i$. By definition, $\nBR_i(s^*)$ minimizes $\tc_i$ if $i$ is the only one allowed to change its strategy, given a condition that $i$ must leave its current facility. Thus, $b_i = \nBR_i(s^*)$. Rewrite $\Delta_i$:

	$$\Delta_i = -\Q_i(s^*) = -(\tc_i(\nBR_i(s^*), s^*_{-i}) - \tc_i(s^*)) = \tc_i(s^*) - \tc_i(b^i).$$

	In \textbf{Case 2}, $\Q_i(s^*) \ge 0$, i.e., $\Q_i(s^*) = \tc_i(\nBR_i(s^*), s^*_{-i}) - \tc_i(s^*) \ge 0$. Similar to the analysis in \textbf{Case 1}, because $\tc_i(\nBR_i(s^*), s^*_{-i}) \ge \tc_i(s^*)$, then it must be $b_i = s^*_i$, so:
	
	$$\Delta_i = 0 = \tc_i(s^*) - \tc_i(b^i).$$

  Rewrite $c(s^*) - \Delta$ as:

	\begin{align*}
		c(s^*) - \Delta &= \sum_{f_k \text{ is open in } s} c(f_k) + \sum_{i \in \mathcal{A}} \tc_i(s^*) - \sum_{i \in \mathcal{A}} \Delta_i \\
		&= \sum_{f_k \text{ is open in } s} c(f_k) +\sum_{i \in \mathcal{A}} \tc_i(s^*) - \sum_{i \in \mathcal{A}}  (\tc_i(s^*) - \tc_i(b^i)) \\
		&= \sum_{f_k \text{ is open in } s} c(f_k) + \sum_{i \in \mathcal{A}} \tc_i(b^i)
	\end{align*}

	Because we assumed there exists $s$ such that $\sum_{i \in \mathcal{A}} \tc_i(b^i)
	\ge  \frac{2}{\tau} \tPhi(s)$, and $\tau \ge 2$, we have that:

	\begin{align*}
		\tau \sum_{i \in \mathcal{A}} \tc_i(b^i) &\ge 2 \tPhi(s)\\
		\tau (\sum_{f_k \text{ is open in } s} c(f_k) + \sum_{i \in \mathcal{A}} \tc_i(b^i) )
		&\ge 2 (\sum_{f_k \text{ is open in } s} c(f_k) + \tPhi(s))\\
		\tau (\sum_{f_k \text{ is open in } s} c(f_k) + \sum_{i \in \mathcal{A}} \tc_i(b^i) )
		&\ge 2 \Phi(s) \\
		\tau (c(s^*) - \Delta)
		&\ge 2 \Phi(s)
		\end{align*}
	
	Starting from assignment $s$, we apply the deviation steps and pricing strategy in Theorem \ref{thm-134-fix3-central-single}. By the analysis in Theorem \ref{thm-134-fix3-central-single}, we will reach a stable state $(\hat{s}, \gamma)$, such that $\Phi(s) \ge \Phi(\hat{s}) \ge \frac{1}{2} c(\hat{s})$, so:

	\begin{align*}
		\tau(c(s^*) - \Delta) &\ge c(\hat{s})\\
		c(s^*) - \Delta &\ge \frac{1}{\tau} c(\hat{s}) \\
		\frac{\Delta}{c(s^*)} &\le 1 - \frac{1}{\tau} \frac{c(\hat{s})}{c(s^*)}
													\le 1 - \frac{1}{\tau} PoS,\\
		\end{align*}
as desired.
\end{proof}

With Lemma~\ref{pay-agents-lemma-s-single}, to prove Theorem~\ref{pay-agents-thm-single}, we only need to find an assignment to satisfy the condition in Lemma~\ref{pay-agents-lemma-s-single} with $\tau = \frac{5}{2}$. We define several assignments $s^0$, $s^1$, $s^2$ below as candidates that may satisfy this condition, and then prove that at least one of them must do so for every instance.

\noindent Define the following state as $s^0$: start from $s^*$, let every agent $i$ switch its strategy to $b_i$, i.e. $s^0_i = b_i$.

We decompose all agents into two groups: $\mathcal{A}_1$ and $\mathcal{A}_2$, such that $\mathcal{A}_1 \cup \mathcal{A}_2 = \mathcal{A}$ and $\mathcal{A}_1 \cap \mathcal{A}_2 = \emptyset$. The agents in $\mathcal{A}_1$ satisfy the following condition: $\forall i, j \in \mathcal{A}_1$, if $b_i \ne b_j$, then it must be the case that $b_i \ne s^*_j$ or $b_j \ne s^*_i$. In other words, it cannot be that by following their best responses, the agents $i$ and $j$ ``switch places", with $i$ moving to $s_j^*$ and $j$ moving to $s_i^*$. The same condition holds for $\mathcal{A}_2$: $\forall i, j \in \mathcal{A}_2$, if $b_i \ne b_j$, then it must be the case that $b_i \ne s^*_j$ or $b_j \ne s^*_i$. There always exist such a decomposition for any instance: For every pair of facilities $f_k,f_\ell$, with $k\leq \ell$, take all the agents $i$ with $s_i^*=f_k$ and $b_i=f_\ell$. Put those in $\mathcal{A}_1$. Similarly, put all agents $j$ with $s_j^*=f_\ell$ and $b_j=f_k$ into $\mathcal{A}_2$. Do this for every pair of facilities, so now $\mathcal{A}_1$ consists of agents who are moving to a higher-numbered facility, and $\mathcal{A}_2$ of agents who are moving to a lower-numbered facility. Finally, for any other agent $i$ such that either $s^*_i = b_i$ or $b_i = \emptyset$ or $s_i^*=\emptyset$, put it in either $\mathcal{A}_1$ or $\mathcal{A}_2$ arbitrarily. The sets $\mathcal{A}_1$ and $\mathcal{A}_2$ which are created clearly satisfy the desired conditions.

Define the following assignment as $s^1$: start from $s^*$ and let every agent $i$ in $\mathcal{A}_1$ switch its strategy to $b_i$, while every agent $i$ in $\mathcal{A}_2$ stays at $s^*_i$. Similarly, define $s^2$ as the assignment that starts from $s^*$ and lets every agent $i$ in $\mathcal{A}_2$ switch its strategy to $b_i$, while every agent $i$ in $\mathcal{A}_1$ stays at $s^*_i$.

Now we will show several lemmas based on the properties of $\mathcal{A}_1$ and $\mathcal{A}_2$.

\begin{lemma}
	\label{pay-agents-lemma-a1-single}
	 For any $\mathcal{A}_1$ and $\mathcal{A}_2$ that satisfy the definition above, these inequalities always hold:

	 $$\sum_{i \in \mathcal{A}_1} \tc_i(b^i) \ge \CC{s^1}{\mathcal{A}_1} + \DC{s^1}{\mathcal{A}_1}{\mathcal{A}_1} + \DC{s^1}{\mathcal{A}_1}{\mathcal{A}_2}$$

	 $$\sum_{i \in \mathcal{A}_2} \tc_i(b^i) \ge \CC{s^2}{\mathcal{A}_2} + \DC{s^2}{\mathcal{A}_2}{\mathcal{A}_2} + \DC{s^2}{\mathcal{A}_1}{\mathcal{A}_2}$$
\end{lemma}

\begin{proof}

	In assignment $s^1$, every agent $i \in \mathcal{A}_1$ switches its strategy to $b_i$, i.e., $b_i = s^1_i$. Also, every agent $j \in \mathcal{A}_2$ stays at $s^*_j$, i.e., $s^1_j = s^*_j$. In assignment $b^i$, agent $i$ is assigned to $b_i$, i.e, $b^i_i = b_i$, and any other agent $j$ stay at $s^*$, i.e., $b^i_j = s^*_j$.\\

	For every $i \in \mathcal{A}_1$:

	$$
	\tc_i(b^i) = \cc(i, b_i) + \sum_{j| b_i \ne s^*_j} \dc(i, j)
	= \cc(i, s^1_i) + \sum_{j| b_i \ne s^*_j} \dc(i, j)
	$$

	Sum up for all $i \in \mathcal{A}_1$:

	\begin{align*}
		\sum_{i \in \mathcal{A}_1} \tc_i(b^i)
		&= \sum_{i \in \mathcal{A}_1}  \cc(i, s^1_i) + \sum_{i \in \mathcal{A}_1} \sum_{j| b_i \ne s^*_j} \dc(i, j) \\
		&\ge \sum_{i \in \mathcal{A}_1}  \cc(i, s^1_i) + \sum_{(i, j)| i, j \in \mathcal{A}_1, b_i \ne s^*_j \lor b_j \ne s^*_i} \dc(i, j) + \sum_{(i, j)| i \in \mathcal{A}_1, j \in \mathcal{A}_2, b_i \ne s^*_j} \dc(i, j)\\
	\end{align*}

	For all agents $i,j \in \mathcal{A}_1$, remember the condition that if $b_i \ne b_j$, then it must be $b_i \ne s^*_j$ or $b_j \ne s^*_i$. This means $\{(i,j)| i, j \in \mathcal{A}_1, s^1_i \ne s^1_j\} = \{(i,j)| i, j \in \mathcal{A}_1, b_i \ne b_j\} \subseteq \{ (i,j)| i, j \in \mathcal{A}_1, b_i \ne s^*_j \lor s^*_i \ne b_j\}$, so we can bound the inequality above by:

	\begin{align*}
		\sum_{i \in \mathcal{A}_1} \tc_i(b^i)
		&\ge \sum_{i \in \mathcal{A}_1}  \cc(i, s^1_i) + \sum_{(i, j)| i, j \in \mathcal{A}_1, b_i \ne s^*_j \lor b_j \ne s^*_i} \dc(i, j) + \sum_{(i, j)| i \in \mathcal{A}_1, j \in \mathcal{A}_2, b_i \ne s^*_j} \dc(i, j)\\
		&\ge \sum_{i \in \mathcal{A}_1}  \cc(i, s^1_i) + \sum_{(i, j)| i, j \in \mathcal{A}_1, s^1_i \ne s^1_j} \dc(i, j) + \sum_{(i, j)| i \in \mathcal{A}_1, j \in \mathcal{A}_2, s^1_i \ne s^1_j} \dc(i, j)\\
		&= \CC{s^1}{\mathcal{A}_1} + \DC{s^1}{\mathcal{A}_1}{\mathcal{A}_1} + \DC{s^1}{\mathcal{A}_1}{\mathcal{A}_2}
	\end{align*}

	The bound for agents in $\mathcal{A}_2$ can be proved similarly.
\end{proof}

By our construction of $\mathcal{A}_1$ and $\mathcal{A}_2$, if there is a pair of agents $(i, j)$ such that $b_i \ne b_j$, $b_i = s^*_j$, and $b_j = s^*_i$, then it must be that one of $i, j$ is in $\mathcal{A}_1$, and the other one is in $\mathcal{A}_2$. We denote the set of such pairs of $(i, j)$ as $\mathcal{Z}$ for convenience. Formally,

	$$\mathcal{Z} = \{(i, j) | i \in \mathcal{A}_1, j \in \mathcal{A}_2, b_i \ne b_j, b_i = s^*_j, b_j = s^*_i\}$$

Note that by definition, for any pair of agents $(i, j) \in \mathcal{Z}$, it must be the case that $i$ and $j$ are not at the same facility in $s^*$. 
Thus, 

	\begin{equation}
		\label{pay-agents-thm-single-eqZ}
		\DC{s^*}{\mathcal{A}_1}{\mathcal{A}_2}
		 = \sum_{(i,j) | i \in \mathcal{A}_1, j \in \mathcal{A}_2, s^*_i \ne s^*_j} \dc(i, j)
		 \ge  \sum_{(i,j) \in \mathcal{Z}} \dc(i, j)
	\end{equation}

	We first show the following lemma to bound $\tPhi(s^0)$ and $\tc(s^0)$ by the sum of $\tc_i(b^i)$ and the total disconnection cost between agents in $\mathcal{Z}$ in $s^*$:

	\begin{lemma}
		\label{pay-agents-lemma-br-phi-single}
		$$\sum_{i \in \mathcal{A}} \tc_i(b^i) + \sum_{(i, j) \in \mathcal{Z}} \dc(i, j) \ge \tPhi(s^0)$$

		$$2 \sum_{i \in \mathcal{A}} \tc_i(b^i) + 2 \sum_{(i, j) \in \mathcal{Z}} \dc(i, j) \ge \tc(s^0)$$
	\end{lemma}

	\begin{proof}

	By lemma \ref{pay-agents-lemma-a1-single}, we know:

	\begin{equation}
		\label{pay-agents-lemma-br-phi-single-eq1}
	\sum_{i \in \mathcal{A}_1} \tc_i(b^i) \ge\CC{s^1}{\mathcal{A}_1} + \DC{s^1}{\mathcal{A}_1}{\mathcal{A}_1} + \DC{s^1}{\mathcal{A}_1}{\mathcal{A}_2}
	\end{equation}

	For agent $i \in \mathcal{A}_1$, $i$'s strategy is $b_i$ in both $s^1$ and $s^0$, so the connection cost to these agents are the same in $s_1$ and $s^0$. Also, for each pair of $(i, j)$ that are both in $\mathcal{A}_1$, the disconnection cost between them is the same in $s^1$ and $s^0$. For a pair $(i, j)$ such that $i \in \mathcal{A}_1$, $j \in \mathcal{A}_2$, we know $i$'s strategy is $b_i$ in both $b^i$ and $s^1$, and $j$ stay at $s^*_j$ in both $b^i$ and $s^1$, so $s^1_i \ne s^1_j$ is equivalent to $b_i \ne s^*_j$. Thus, Inequality~\ref{pay-agents-lemma-br-phi-single-eq1} is equivalent to:

	$$\sum_{i \in \mathcal{A}_1} \tc_i(b^i) \ge \CC{s^0}{\mathcal{A}_1} + \DC{s^0}{\mathcal{A}_1}{\mathcal{A}_1} + \sum_{(i,j)| i \in \mathcal{A}_1, j \in \mathcal{A}_2, b_i \ne s^*_j} d(i, j)$$

	Similarly, we can get the following bounds for agents in $\mathcal{A}_2$:

	$$\sum_{i \in \mathcal{A}_2} \tc_i(b^i) \ge \CC{s^0}{\mathcal{A}_2} + \DC{s^0}{\mathcal{A}_2}{\mathcal{A}_2} + \sum_{(i,j)| i \in \mathcal{A}_1, j \in \mathcal{A}_2, b_j \ne s^*_i} d(i, j)$$

	Summing them up:
	\begin{equation}
		\label{pay-agents-lemma-br-phi-single-eq2}
	\sum_{i \in \mathcal{A}} \tc_i(b^i)
	\ge \CC{s^0}{\mathcal{A}} + \DC{s^0}{\mathcal{A}_1}{\mathcal{A}_1} + \DC{s^0}{\mathcal{A}_2}{\mathcal{A}_2}
	+ \sum_{(i,j)| i \in \mathcal{A}_1, j \in \mathcal{A}_2, b_i \ne s^*_j \lor s^*_i \ne b_j} d(i, j)
	\end{equation}

	By definition, $\tPhi(s^0)$ equals:

	\begin{align*}
	\tPhi(s^0) &= \CC{s^0}{\mathcal{A}} + \DC{s^0}{\mathcal{A}}{\mathcal{A}}\\
	&= \CC{s^0}{\mathcal{A}} + \DC{s^0}{\mathcal{A}_1}{\mathcal{A}_1} + \DC{s^0}{\mathcal{A}_2}{\mathcal{A}_2} + \DC{s^0}{\mathcal{A}_1}{\mathcal{A}_2}
	\end{align*}

	The only difference between $\tPhi(s^0)$ and the right hand side of Inequality~\ref{pay-agents-lemma-br-phi-single-eq2} is the last term. We further decompose $\DC{s^0}{\mathcal{A}_1}{\mathcal{A}_2}$:

	\begin{align}
		\DC{s^0}{\mathcal{A}_1}{\mathcal{A}_2}
		&=  \sum_{(i,j)| i \in \mathcal{A}_1, j \in \mathcal{A}_2, b_i \ne b_j} \dc(i, j) \nonumber \\
		&=  \sum_{(i,j)| i \in \mathcal{A}_1, j \in \mathcal{A}_2, b_i \ne b_j, b_i \ne s^*_j \lor b_j \ne s^*_i} \dc(i, j)
		+  \sum_{(i,j)| i \in \mathcal{A}_1, j \in \mathcal{A}_2, b_i \ne b_j, b_i = s^*_j, b_j = s^*_i} \dc(i, j)\nonumber \\
		&\le  \sum_{(i,j)| i \in \mathcal{A}_1, j \in \mathcal{A}_2, b_i \ne s^*_j \lor b_j \ne s^*_i} \dc(i, j)
		+ \sum_{(i,j) \in \mathcal{Z}} \dc(i, j) \label{pay-agents-lemma-br-phi-single-eq3}
	\end{align}

	Put it back to the definition of $\tPhi(s^0)$:
	\begin{align*}
		\tPhi(s^0) &\le \CC{s^0}{\mathcal{A}} + \DC{s^0}{\mathcal{A}_1}{\mathcal{A}_1} + \DC{s^0}{\mathcal{A}_2}{\mathcal{A}_2}\\
		&+ \sum_{(i,j)| i \in \mathcal{A}_1, j \in \mathcal{A}_2, b_i \ne s^*_j \lor s^*_i \ne b_j} d(i, j)
		+ \sum_{(i,j) \in \mathcal{Z}} d(i, j)
	\end{align*}

	Combine the inequality above with Inequality~\ref{pay-agents-lemma-br-phi-single-eq2}, then we have proved the bound for $\tPhi(s^0)$, as desired:

	$$\sum_{i \in \mathcal{A}} \tc_i(b^i) + \sum_{(i, j) \in \mathcal{Z}} d(i, j) \ge \tPhi(s^0)$$

	We get the bound for $\tc(s^0)$ similarly. By the definition of $\tc(s^0)$ and Inequality~\ref{pay-agents-lemma-br-phi-single-eq3}:
	\begin{align*}
	\tc(s^0) &= \CC{s^0}{\mathcal{A}} + 2 \DC{s^0}{\mathcal{A}_1}{\mathcal{A}_1} + 2 \DC{s^0}{\mathcal{A}_2}{\mathcal{A}_2} + 2 \DC{s^0}{\mathcal{A}_1}{\mathcal{A}_2}\\
	&\le \CC{s^0}{\mathcal{A}} + 2\DC{s^0}{\mathcal{A}_1}{\mathcal{A}_1} + 2\DC{s^0}{\mathcal{A}_2}{\mathcal{A}_2}\\
	&+ 2\sum_{(i,j)| i \in \mathcal{A}_1, j \in \mathcal{A}_2, b_i \ne s^*_j \lor s^*_i \ne b_j} d(i, j)
	+ 2 \sum_{(i,j) \in \mathcal{Z}} d(i, j)
	\end{align*}

	Combine the inequality above with Inequality~\ref{pay-agents-lemma-br-phi-single-eq2}, we have proved the bound for $\tc(s^0)$:

	$$2 \sum_{i \in \mathcal{A}} \tc_i(b^i) + 2 \sum_{(i, j) \in \mathcal{Z}} d(i, j) \ge \tc(s^0).$$
  \end{proof}

Using the previous lemmas, we are now ready to prove Theorem~\ref{pay-agents-thm-single}.

\begin{proof}
	By Lemma~\ref{pay-agents-lemma-s-single}, to prove Theorem~\ref{pay-agents-thm-single}, we only need to prove there exists $s$, such that:

	\begin{equation}
		\label{pay-agents-thm-single-eq-goal}
	\sum_{i \in \mathcal{A}} \tc_i(b^i)
	\ge  \frac{4}{5} \tPhi(s).
	\end{equation}

	We prove the conclusion by contradiction. Suppose there does not exist any state $s$ to make Inequality \ref{pay-agents-thm-single-eq-goal} hold.

	By our assumption, Inequality \ref{pay-agents-thm-single-eq-goal} does not hold for $s^1$:

	\begin{align}
	\sum_{i \in \mathcal{A}} \tc_i(b^i)
	&< \frac{4}{5} \tPhi(s^1) \nonumber \\
	\sum_{i \in \mathcal{A}_1} \tc_i(b^i) + \sum_{i \in \mathcal{A}_2} \tc_i(b^i)
	&< \frac{4}{5} (\CC{s^1}{\mathcal{A}} + \DC{s^1}{\mathcal{A}}{\mathcal{A}}) \nonumber \\
	\sum_{i \in \mathcal{A}_1} \tc_i(b^i) + \sum_{i \in \mathcal{A}_2} \tc_i(b^i)
	&< \frac{4}{5} (\CC{s^1}{\mathcal{A}_1} +  \CC{s^1}{\mathcal{A}_2} \\
		&+ \DC{s^1}{\mathcal{A}_1}{\mathcal{A}_1} + \DC{s^1}{\mathcal{A}_2}{\mathcal{A}_2} + \DC{s^1}{\mathcal{A}_1}{\mathcal{A}_2} )
	\label{pay-agents-thm-single-eq2}
	\end{align}

  By Lemma \ref{pay-agents-lemma-a1-single}:
	$$\frac{4}{5} \sum_{i \in \mathcal{A}_1} \tc_i(b^i) \ge \frac{4}{5}(\CC{s^1}{\mathcal{A}_1} + \DC{s^1}{\mathcal{A}_1}{\mathcal{A}_1} + \DC{s^1}{\mathcal{A}_1}{\mathcal{A}_2})$$

	Subtract both sides from Inequality \ref{pay-agents-thm-single-eq2}, we get:
	\begin{equation}
		\label{pay-agents-thm-single-eq3}
		\frac{1}{5} \sum_{i \in \mathcal{A}_1} \tc_i(b^i) + \sum_{i \in \mathcal{A}_2} \tc_i(b^i)
		< \frac{4}{5}(\CC{s^1}{\mathcal{A}_2} + \DC{s^1}{\mathcal{A}_2}{\mathcal{A}_2})
	\end{equation}

	In assignment $s^1$, for any agent $i \in \mathcal{A}_2$, it stays at the same facility as in $s^*$, so the connection cost to these agents are the same as in $s^*$. Also, for a pair of agents $(i, j)$ both in $\mathcal{A}_2$, because they both stay at the facility in $s^*$, the disconnection cost between is also the same as in $s^*$. So we can rewrite Inequality \ref{pay-agents-thm-single-eq3} as:

	\begin{equation}
		\label{pay-agents-thm-single-eq4}
		\frac{1}{5} \sum_{i \in \mathcal{A}_1} \tc_i(b^i) + \sum_{i \in \mathcal{A}_2} \tc_i(b^i)
		< \frac{4}{5}(\CC{s^*}{\mathcal{A}_2} + \DC{s^*}{\mathcal{A}_2}{\mathcal{A}_2})
	\end{equation}

	Similarly, by the assumption that Inequality \ref{pay-agents-thm-single-eq-goal} does not hold for $s^2$, we can get:
	\begin{equation}
		\label{pay-agents-thm-single-eq5}
		\frac{1}{5} \sum_{i \in \mathcal{A}_2} \tc_i(b^i) + \sum_{i \in \mathcal{A}_1} \tc_i(b^i)
		< \frac{4}{5}(\CC{s^*}{\mathcal{A}_1} + \DC{s^*}{\mathcal{A}_1}{\mathcal{A}_1})
	\end{equation}

	Sum up Inequality \ref{pay-agents-thm-single-eq4} and \ref{pay-agents-thm-single-eq5}:
	\begin{align}
		\frac{6}{5} \sum_{i \in \mathcal{A}} \tc_i(b^i)	&< \frac{4}{5}(\CC{s^*}{\mathcal{A}_1} + \DC{s^*}{\mathcal{A}_1}{\mathcal{A}_1} + \CC{s^*}{\mathcal{A}_2} + \DC{s^*}{\mathcal{A}_2}{\mathcal{A}_2}) \nonumber \\
		\frac{3}{2} \sum_{i \in \mathcal{A}} \tc_i(b^i)	&< \CC{s^*}{\mathcal{A}} + \DC{s^*}{\mathcal{A}_1}{\mathcal{A}_1} + \DC{s^*}{\mathcal{A}_2}{\mathcal{A}_2} \label{pay-agents-thm-single-eq6}
	\end{align}

	By Lemma ~\ref{pay-agents-lemma-br-phi-single}:

	$$\sum_{i \in \mathcal{A}} \tc_i(b^i) + \sum_{(i, j) \in \mathcal{Z}} \dc(i, j) \ge \tPhi(s^0)$$

	And we assume there is no assignment $s$ that makes Inequality~\ref{pay-agents-thm-single-eq-goal} hold, so:

	$$\sum_{i \in \mathcal{A}} \tc_i(b^i) <  \frac{4}{5} \tPhi(s^0) $$

	Then combine the two inequalities above, we get $\sum_{(i, j) \in \mathcal{Z}} \dc(i, j) > \frac{1}{4} \sum_{i \in \mathcal{A}} \tc_i(b^i)$.

	We add $\frac{1}{2} \sum_{i \in \mathcal{A}} \tc_i(b^i)$ and $2\sum_{(i, j) \in \mathcal{Z}} \dc(i, j)$ to the left and right hand sides of Inequality~\ref{pay-agents-thm-single-eq6}:

	$$2 \sum_{i \in \mathcal{A}} \tc_i(b^i)
	< \CC{s^*}{\mathcal{A}} + \DC{s^*}{\mathcal{A}_1}{\mathcal{A}_1} + \DC{s^*}{\mathcal{A}_2}{\mathcal{A}_2}
	+ 2\sum_{(i, j) \in \mathcal{Z}} \dc(i, j)$$

	Remember from Inequality~\ref{pay-agents-thm-single-eqZ}, we know $\DC{s^*}{\mathcal{A}_1}{\mathcal{A}_2} \ge \sum_{(i,j) \in \mathcal{Z}} \dc(i, j)$, so:

	\begin{align*}
	2 \sum_{i \in \mathcal{A}} \tc_i(b^i)
	&< \CC{s^*}{\mathcal{A}} + \DC{s^*}{\mathcal{A}_1}{\mathcal{A}_1} + \DC{s^*}{\mathcal{A}_2}{\mathcal{A}_2}
	+ 2\DC{s^*}{\mathcal{A}_1}{\mathcal{A}_2}\\
	&\le \CC{s^*}{\mathcal{A}} + 2\DC{s^*}{\mathcal{A}_1}{\mathcal{A}_1} + 2\DC{s^*}{\mathcal{A}_2}{\mathcal{A}_2}
	+ 2\DC{s^*}{\mathcal{A}_1}{\mathcal{A}_2}\\
	&= \tc(s^*).
	\end{align*}

	It is easy to see if $2 \sum_{i \in \mathcal{A}} \tc_i(b^i)	\ge \tc(s^1)$ or $2 \sum_{i \in \mathcal{A}} \tc_i(b^i)	\ge \tc(s^2)$, then there exists an assignment better than $s^*$. We will get more conditions from this contradiction, and then finish our proof.

	If $2 \sum_{i \in \mathcal{A}} \tc_i(b^i)	\ge \tc(s^1)$, by above we get $\tc(s^1) < \tc(s^*)$. And every facility that is open in $s^1$ must also be open in $s^*$: this is because in $s^1$, every agent $i$ in $\mathcal{A}_1$ switches its strategy to $b_i$, and $b_i$ is open in $s^*$ by definition. So the total facility cost in $s^1$ is no larger than that in $s^*$, i.e., $\rc(s^1) < \rc(s^*)$, which contradicts the fact that $s^*$ is the optimal assignment. So we know the assumption that $2 \sum_{i \in \mathcal{A}} \tc_i(b^i)	\ge \tc(s^1)$ must be false, so it must be $2 \sum_{i \in \mathcal{A}} \tc_i(b^i) < \tc(s^1)$:

	\begin{align}
	2 \sum_{i \in \mathcal{A}} \tc_i(b^i)	&< \tc(s^1) \nonumber \\
	2 \sum_{i \in \mathcal{A}_1} \tc_i(b^i) +  2 \sum_{i \in \mathcal{A}_2} \tc_i(b^i)
	&< \CC{s^1}{\mathcal{A}_1} + \CC{s^1}{\mathcal{A}_2} \nonumber \\
	&+ 2 \DC{s^1}{\mathcal{A}_1}{\mathcal{A}_1} + 2 \DC{s^1}{\mathcal{A}_2}{\mathcal{A}_2} + 2 \DC{s^1}{\mathcal{A}_1}{\mathcal{A}_2} \label{pay-agents-thm-single-eq7}
	\end{align}

	By Lemma~\ref{pay-agents-lemma-a1-single}:

	\begin{align}
	\sum_{i \in \mathcal{A}_1} \tc_i(b^i)
	&\ge \CC{s^1}{\mathcal{A}_1} + \DC{s^1}{\mathcal{A}_1}{\mathcal{A}_1} + \DC{s^1}{\mathcal{A}_1}{\mathcal{A}_2} \nonumber \\
	2\sum_{i \in \mathcal{A}_1} \tc_i(b^i)
	&\ge \CC{s^1}{\mathcal{A}_1} + 2\DC{s^1}{\mathcal{A}_1}{\mathcal{A}_1} + 2\DC{s^1}{\mathcal{A}_1}{\mathcal{A}_2} \label{pay-agents-thm-single-eq8}
	\end{align}

	Combine Inequality \ref{pay-agents-thm-single-eq7} and \ref{pay-agents-thm-single-eq8}:
	$$2 \sum_{i \in \mathcal{A}_2} \tc_i(b^i)
	< \CC{s^1}{\mathcal{A}_2} + 2\DC{s^1}{\mathcal{A}_2}{\mathcal{A}_2} $$

	In assignment $s^1$, the agents in $\mathcal{A}_2$ stay at their facility in $s^*$, so the connection cost to these agents are the same as in $s^*$. Also, for two agents $(i, j)$ both in $\mathcal{A}_2$, the disconnection cost between them is also the same as in $s^*$. So the inequality above is equivalent to:

	$$2 \sum_{i \in \mathcal{A}_2} \tc_i(b^i)
	< \CC{s^*}{\mathcal{A}_2} + 2\DC{s^*}{\mathcal{A}_2}{\mathcal{A}_2} $$

	Similarly, if $2 \sum_{i \in \mathcal{A}} \tc_i(b^i)	\ge \tc(s^2)$, we can get $\tc(s^2) < \tc(s^*)$, by the same argument as above, we get:
	$$2 \sum_{i \in \mathcal{A}_1} \tc_i(b^i)
	< \CC{s^*}{\mathcal{A}_1} + 2\DC{s^*}{\mathcal{A}_1}{\mathcal{A}_1} $$

	Summing them up:
	\begin{align*}
		2 \sum_{i \in \mathcal{A}_2} \tc_i(b^i) + 2 \sum_{i \in \mathcal{A}_2} \tc_i(b^i)
	&< \CC{s^*}{\mathcal{A}_2} + 2\DC{s^*}{\mathcal{A}_2}{\mathcal{A}_2} + \CC{s^*}{\mathcal{A}_1} + 2\DC{s^*}{\mathcal{A}_1}{\mathcal{A}_1} \\
	2 \sum_{i \in \mathcal{A}} \tc_i(b^i)
	&< \CC{s^*}{\mathcal{A}} + 2\DC{s^*}{\mathcal{A}_1}{\mathcal{A}_1} + 2\DC{s^*}{\mathcal{A}_2}{\mathcal{A}_2}
	\end{align*}

	Add $2\DC{s^*}{\mathcal{A}_1}{\mathcal{A}_2}$ to both sides:

	\begin{equation*}
	2 \sum_{i \in \mathcal{A}} \tc_i(b^i) + 2\DC{s^*}{\mathcal{A}_1}{\mathcal{A}_2} < \tc(s^*)
	\end{equation*}

	Finally, combine the inequality above with Lemma~\ref{pay-agents-lemma-br-phi-single} and Inequality~\ref{pay-agents-thm-single-eqZ}:

	$$\tc(s^0) \le 2 \sum_{i \in \mathcal{A}} \tc_i(b^i) + 2 \sum_{(i, j) \in \mathcal{Z}} \dc(i, j)
	\le 2 \sum_{i \in \mathcal{A}} \tc_i(b^i) + 2 \DC{s^*}{\mathcal{A}_1}{\mathcal{A}_2}
	< \tc(s^*)$$

	Remember $s^0$ is the assignment such that every agent $i$ in $\mathcal{A}_1$ switches its strategy to $b_i$, and no agent switches to a closed facility in $s^*$ by the definition of $b_i$. So the total facility cost in $s^0$ is no more than that in $s^*$. Thus we get $\rc(s^0) < \rc(s^*)$, which contradicts the fact that $s^*$ is the optimal assignment. Thus, the assumption that no assignment $s^0, s^1, s^2$ satisfies Inequality~\ref{pay-agents-thm-single-eq-goal} is false, and this proves the theorem.
\end{proof}

\section{Multiple Facilities per Agent}
\label{sec-multi}
In this section, we consider the case that each agent is allowed to use multiple facilities. Most of our results still hold in this setting, but with a constraint on possible deviations: when agents switch their strategies, they are only allowed to drop from at most one facility each time, although they can join as many new facilities as they want to. Most of our notation and proofs are similar to the single facility setting; we include them here for completeness.

\subsection{Model in the Multiple Facilities Setting}
\label{sec-model-multi}
We are given a set of $m$ facilities $\mathcal{F} = \{f_1, f_2, \dots, f_m\}$ and a set of $n$ agents $\mathcal{A} = \{1, 2, \dots, n\}$. An agent $i$ can use any facility $f_k$ by paying a connection cost $\cc(i, f_k)$. A pair of agents $(i, j)$ can form connections through facility $f_k$ if they are both using $f_k$. However, if $i$ and $j$ do not share any facility that they are using, then both of them are charged a disconnection cost $\dc(i, j)$. We say a facility $f_k$ is open if and only if there exists an agent using it. There is a fixed facility cost of $c(f_k) \ge 0$ for any open facility $f_k$.

 A facility assignment $s = \{s_1, s_2, \dots, s_n\}$ assigns the set of facilities that each agent uses: $s_i$ denotes the set of facilities that agent $i$ uses in assignment $s$. In the case that agent $i$ does not use any facility, define $s_i = \emptyset$ and $\cc(i, s_i) = 0$. A pricing strategy $\gamma = \{\gamma_1, \gamma_2, \dots, \gamma_n\}$ assigns the price for using each facility $f_k$ to every agent $i$.  $\gamma_i(f_k)$ is a non-negative number that denotes agent $i$'s share of the facility cost for using $f_k$. $(s, \gamma)$ is a state with assignment $s$ and pricing strategy $\gamma$. Note that agent $i$ only pays its share of the facility cost to $f_k$ if $i$ uses $f_k$, i.e., $\gamma_i(f_k) > 0$ only if $f_k \in s_i$. \\

The total cost of agent $i$ in a state $(s, \gamma)$ is the sum of the following three parts:

\begin{enumerate}
	\item If $i$ uses facility $f_k$, then there is a connection cost $\cc(i, f_k)$ to $i$.
	\item For each agent $j$ that do not share any facility with agent $i$, i.e., $s_i \cap s_j = \emptyset$, there is a disconnection cost $\dc(i, j)$ to both $i$ and $j$.
	\item If $i$ uses facility $f_k$, then there is a facility cost $\gamma_i(f_k)$ to $i$.
\end{enumerate}

We denote the total cost of agent $i$ as $\rc_i(s, \gamma)$. Sum up the three types of cost mentioned above:

$$\rc_i(s, \gamma) = \sum_{f_k \in s_i} (\cc(i, f_k) + \gamma_i(f_k)) + \sum_{j | s_i \cap s_j= \emptyset} \dc(i, j)$$

We denote the cost of agent $i$ without facility cost as $\tc_i(s)$:

$$\tc_i(s) = \sum_{f_k \in s_i} \cc(i, f_k) + \sum_{j | s_i \cap s_j= \emptyset} \dc(i, j)$$

The total social cost of a state $(s, \gamma)$ equals the sum of $\rc_i(s, \gamma)$, plus the total cost of all open facilities. For each facility $f_k$, the cost is $c(f_k)$ minus the sum of $\gamma_i(f_k)$ of each agent using $f_k$, i.e., $c(f_k) - \sum_{i|s_i = f_k} \gamma_i(f_k)$. The sum of $\gamma_i(f_k)$ cancels out, and the total social cost is actually the sum of $\tc_i(s)$ plus the sum of $c(f_k)$ of open facilities:

\begin{align*}
	c(s) &= \sum_{f_k \in \mathcal{F}, f_k \text{ is open}} c(f_k) +
					\sum_{i \in \mathcal{A}} ( \sum_{f_k \in s_i} \cc(i, f_k) + \sum_{j | s_i \cap s_j= \emptyset} \dc(i, j))\\
			 &= \sum_{f_k \in \mathcal{F}, f_k \text{ is open}} c(f_k)
					+ \sum_{i \in \mathcal{A}} \sum_{f_k \in s_i} \cc(i, f_k) + 2\sum_{(i, j) | s_i \cap s_j= \emptyset} \dc(i, j)
	\end{align*}

We consider $(i, j)$ as an unordered pair, therefore in $\sum_{i \in \mathcal{A}} \sum_{j | s_i \cap s_j= \emptyset} \dc(i, j)$, each unordered pair $(i, j)$ that $s_i \cap s_j= \emptyset$ is counted twice.

In this paper, we study the game where each agent's goal is to minimize its total social cost, and the central coordinator's goal is to find a budget balanced and stable state $(s, \gamma)$ that (approximately) minimizes the total social cost. A state is budget balanced if each facility $f_k$ is fully paid with the facility cost $c(f_k)$, formally:

\begin{definition}
	A state $(s, \gamma)$ is \textbf{budget balanced} if for each facility $f_k$, $\sum_{i | f_k \in s_i} \gamma_i(f_k) = c(f_k)$.
\end{definition}

In the multiple facilities setting, when agents deviate from their current strategies, they are only allowed to drop from at most one facility each time, formally:

\begin{definition}
	\label{def-valid-dev-multi}
	Given a state $(s, \gamma)$, $s'_i$ is agent $i$'s valid deviation if $|s_i\backslash s'_i| \le 1$.
\end{definition}

Before defining the stability of a state, we first define an agent's best response and stability. In this paper, an agent assumes it will be charged 0 facility cost for joining a new facility; see Section \ref{sec-model} for why this actually makes our results stronger than any other assumption, since if a stable solution exists with this assumption, then it will still be stable if agents assume they would have to pay a different cost. In other words, if an agent is stable when assuming it will be charged 0 for joining some facilities, then it would also be stable with a higher cost.

\begin{definition}
	Given a state $(s, \gamma)$, $s_i'$ is agent $i$'s \textbf{best response} if for any valid deviation $s_i''$,

	$$\tc_i(s_i', s_{-i}) + \sum_{f_k \in s_i \cap s_i'} \gamma_i(f_k) \le \tc_i(s_i'', s_{-i}) + \sum_{f_k \in s_i \cap s_i''} \gamma_i(f_k)$$

	We denote $i$'s best response at state $(s, \gamma)$ as $\BR_i(s, \gamma)$.
\end{definition}

Then we define the stability of an agent:
\begin{definition}
	Agent $i$ is stable at state $(s, \gamma)$ if for any valid deviation $s_i'$:
	$$\rc_i(s, \gamma) \le \tc_i(s_i', s_{-i}) + \sum_{f_k \in s_i \cap s_i'} \gamma_i(f_k)$$
\end{definition}

In other words, agent $i$ is stable at $(s, \gamma)$ if $s_i$ is $i$'s best response at $(s, \gamma)$.

We define a state $(s, \gamma)$ to be stable if it is budget balanced, and every agent is stable:
\begin{definition}
	A state $(s, \gamma)$ is \textbf{stable} if it is budget balanced, and for each agent $i$, for any valid deviation $s_i'$:
	$$\rc_i(s, \gamma) \le \tc_i(s_i', s_{-i}) + \sum_{f_k \in s_i \cap s_i'} \gamma_i(f_k)$$
\end{definition}

\subsection{Pricing Strategies and Stability}
\label{sec-pricing-stable-multi}
Remember a stable state $(s, \gamma)$ is budget balanced, and every agent $i$ is stable. Suppose there is no other constraint on the pricing strategy, then we ask the following question in order to find a budget balanced state: in an assignment $s$, how much facility cost can we charge an agent while keeping it stable? Consider an agent $i$ such that $f_k \in s_i$. To see how much we can charge $i$ without causing it to deviate away from $f_k$, we first define a special type of best response: with an assignment $s$, let $\nBR_i(s, f_k)$ denote $i$'s best response with regard to $\tc_i(s)$, given $i$ is forced to stop using $f_k$ (and forbidden to join $f_k$ again). If $i$ does not use any facility in $s$, $\nBR_i(s, \emptyset)$ is just $i$'s best response with regard to $\tc_i(s)$. $\nBR_i(s, f_k)$ is agent $i$'s ``next best choice'' if $i$ is forced to stop using $f_k$. Intuitively, the ``value'' of facility $f_k$ to agent $i$ is how much $i$'s cost would increase if $i$ is forced to leave $f_k$ and join the next best choice  $\nBR_i(s, f_k)$. If there are multiple strategies that all satisfy the definition of $\nBR_i(s, f_k)$, we choose an arbitrary one except in one case: we never choose a facility set that contains a closed facility in $s$ as $\nBR_i(s, f_k)$. We can always do this because if there exists such strategy $s_i'$, that $s_i'$ contains a closed facility $f_k'$ in $s$, then compare $\tc_i(s_i'-\{f_k'\}, s_{-i})$ with $\tc_i(s_i', s_{-i})$. The connection cost in $\tc_i(s_i'-\{f_k'\}, s_{-i})$ is no more than that in $\tc_i(s_i', s_{-i})$, and the disconnection cost is no less than that in $\tc_i(s_i', s_{-i})$, because $i$ would be the only agent using $f_k'$ in $s_i'$. So it must be $\tc_i(s_i'-\{f_k'\}, s_{-i}) \le \tc_i(s_i', s_{-i})$, then we can remove all facilities that are open in $s_i'$ but closed in $s$ to reach an assignment $\hat{s}_i$ such that $\tc_i(\hat{s}_i, s_{-i}) \le \tc_i(s_i', s_{-i})$. We define $\nBR_i(s) = \hat{s}_i$ in this case.

Note that some agents might be unstable even with 0 facility cost. Thus, we also consider the case that agents need to receive payments to be stable. If agent $i$ uses facility $f_k$ in $s$, let $\Delta_i(f_k)$ denote the payment that agent $i$ receives if it does not deviate away from $f_k$ at state $(s, \gamma)$, and denote the total payments as $\Delta = \sum_i \sum_{f_k \in s_i} \Delta_i (f_k)$. Then we define the stability with payments as follows: a state $(s, \gamma)$ with payments $\Delta$ is stable if it is budget balanced, and for each agent $i$, for any valid deviation $s_i'$ of agent $i$:

$$\rc_i(s, \gamma) - \sum_{f_k \in s_i} \Delta_i(f_k) \le \tc_i(s_i', s_{-i}) + \sum_{f_k \in s_i \cap s_i'} \gamma_i(f_k) - \sum_{f_k \in s_i \cap s_i'} \Delta_i(f_k)$$

For every agent $i$ using facility $f_k$ in assignment $s$, define $\Q_i(s, f_k) = \tc_i(\nBR_i(s, f_k), s_{-i}) - \tc_i(s)$. We will show a condition that guarantees agent $i$ is stable in the following lemma.

\begin{lemma}
	\label{lemma-Q-multi}
Given any state $(s, \gamma)$, agent $i$ is stable with payments if it satisfies the following two conditions:

\begin{enumerate}

\item For any $s_i'$ such that $s_i \subseteq s_i'$ (when deviating to $s_i'$, $i$ does not leave any facility in $s_i$), $\tc_i(s) \le \tc_i(s_i', s_{-i})$.

\item $\forall f_k \in s_i$, $\gamma_i(f_k) - \Delta_i(f_k) \le \Q_i(s, f_k)$.

\end{enumerate}

\end{lemma}

\begin{proof}

	First consider the first type of deviation $s_i'$, in which agent $i$ does not leave any facility, but might join some new ones. In other words, the set $\{f_k | f_k \in s_i\}$ is equivalent to $\{f_k | f_k \in s_i \cap s_i'\}$. With the given condition:

	\begin{align*}
	\tc_i(s, \gamma) &\le \tc_i(s_i', s_{-i})\\
	\tc_i(s, \gamma) + \sum_{f_k \in s_i} \gamma_i(f_k) - \sum_{f_k \in s_i} \Delta_i(f_k)
	&\le \tc_i(s_i', s_{-i}) + \sum_{f_k \in s_i \cap s_i'} \gamma_i(f_k) - \sum_{f_k \in s_i \cap s_i'} \Delta_i(f_k) \\
	\rc_i(s, \gamma) - \sum_{f_k \in s_i} \Delta_i(f_k)
	&\le \tc_i(s_i', s_{-i}) + \sum_{f_k \in s_i \cap s_i'} \gamma_i(f_k) - \sum_{f_k \in s_i \cap s_i'} \Delta_i(f_k) \\
  \end{align*}

	Thus, agent $i$ would not deviate to $s_i'$ such that $s_i \subseteq s_i'$.

	Then consider the type of deviation that agent $i$ deviates to $s_i'$ such that $s_i \nsubseteq s_i'$. Because agent $i$ can only leave one facility in each deviation, we will consider the deviation includes leaving each facility separately. For any facility $f_k \in s_i$, with $\gamma_i(f_k) - \Delta_i(f_k) \le \Q_i(s, f_k)$, by the definition of agent $i$'s total cost:
	\begin{align*}
		&\rc_i(s, \gamma) - \sum_{f_k' \in s_i} \Delta_i(f_k')\\
		&= \tc_i(s) + \sum_{f_k' \in s_i} \gamma_i(f_k') - \sum_{f_k' \in s_i} \Delta_i(f_k')\\
		&= \tc_i(s) + \sum_{f_k' | f_k' \in s_i, f_k' \ne f_k} \gamma_i(f_k') + \gamma_i(f_k) - \sum_{f_k' | f_k' \in s_i, f_k' \ne f_k} \Delta_i(f_k') - \Delta_i(f_k)\\
		&\le \tc_i(s) + \Q_i(s, f_k) + \sum_{f_k' |f_k' \in s_i, f_k' \ne f_k} \gamma_i(f_k') - \sum_{f_k' | f_k' \in s_i, f_k' \ne f_k} \Delta_i(f_k') \\
		&= \tc_i(s) + \tc_i(\nBR_i(s, f_k), s_{-i}) - \tc_i(s) + \sum_{f_k' |f_k' \in s_i, f_k' \ne f_k} \gamma_i(f_k') - \sum_{f_k' | f_k' \in s_i, f_k' \ne f_k} \Delta_i(f_k') \\
		&= \tc_i(\nBR_i(s, f_k), s_{-i}) + \sum_{f_k' |f_k' \in s_i, f_k' \ne f_k} \gamma_i(f_k') - \sum_{f_k' | f_k' \in s_i, f_k' \ne f_k} \Delta_i(f_k')
		\end{align*}

	By the definition of $\nBR_i(s, f_k)$, for any valid deviation $s_i'$ with $f_k \notin s_i'$:
	$$\tc_i(\nBR_i(s, f_k), s_{-i}) \le \tc_i(s_i', s_{-i})$$

	Also, because $f_k \notin s_i'$ and $i$ can leave at most one facility at a time, so $\{f_k' | f_k' \in s_i, f_k' \ne f_k\} = \{f_k' | f_k' \in s_i \cap s_i'\}$. Therefore, $i$ would not deviate to any $s_i'$ that $i$ leaves at most one facility compared to $s_i$. So every $i$ is stable to any valid deviation.
\end{proof}

In this paper, the default setting is that agents do not receive payments ($\Delta_i = 0$), but we do consider the cases that agents are allowed to be paid by their neighbors in Section~\ref{sec-paid-peering-multi}.

\subsection{Facility cost $c(f_k) = 0$ for every $f_k$}
In this section, we consider the case that there is no facility cost, i.e. $\forall k, c(f_k) = 0$. Set the pricing strategy to be $\gamma_i(f_k) = 0$ for any agent $i$ and facility $f_k$, then $\gamma$ is budget balanced. In this setting, for any agent $i$ in assignment $s$, we have $\rc_i(s, \gamma) = \tc_i(s)$. A state $s$ is stable if for each agent $i$ and assignment $s_i'$, $\tc_i(s) \le \tc_i(s_i', s_{-i})$.

Define potential function $\tPhi(s)$ as:

$$\tPhi(s) = \sum_{i \in \mathcal{A}} \sum_{f_k \in s_i} \cc(i, f_k) + \sum_{(i,j) | s_i \cap s_j= \emptyset} \dc(i, j)$$

When an agent $i$ switches its strategy from $s_i$ to $s_i'$, it is easy to see that the change of $i$'s cost is captured exactly by the change of $\tPhi(s)$, so $\tPhi(s)$ is an exactly potential function.

\begin{theorem}
\label{thm-14-multi}
	If $\forall k, c(f_k) = 0$, then PoS is at most 2 and this bound is tight.
\end{theorem}

\begin{proof}
	The proof is almost the same as Theorem\ref{thm-14-single}, with the potential function $\tPhi(s)$ defined above. We include the proof of this theorem as well as other theorems in this section in the appendix for completeness, and only explain the difference between the single and multiple facilities setting in the main body of this paper.
\end{proof}

\begin{theorem}
\label{thm-14-poa-multi}
  PoA is unbounded in our setting.
\end{theorem}

\begin{proof}
	 The example in Theorem~\ref{thm-14-poa-single} also works in the multiple facilities setting.
\end{proof}

\subsection{Non-zero Facility Costs: Price of Stability}
\label{sec-pos-multi}
Same as in Section \ref{sec-pos-single}, we consider the general case that the facility cost $c(f_k)$ is a fixed constant when $f_k$ is open. The only differences from Section \ref{sec-pos-single} is that each agent is allowed to use multiple facilities, and the valid deviation of an agent is defined in Definition~\ref{def-valid-dev-multi}.

We define potential function $\Phi(s)$:

\begin{align*}
	\Phi(s) &= \sum_{f_k | f_k \text{ is open}} c(f_k) + \sum_{i \in \mathcal{A}} \sum_{f_k \in s_i} \cc(i, f_k) + \sum_{(i,j) | s_i \cap s_j= \emptyset} \dc(i, j)\\
	&=  \sum_{f_k | f_k \text{ is open}} c(f_k) + \tPhi(s)
\end{align*}

We will show that $\Phi(s)$ is an ordinal potential function when any agent $i$ switches to a strategy that decreases its cost without the facility cost.

\begin{lemma}
	\label{lemma-phi-basic-multi}
	In an assignment $s$, if any agent $i$ switch its strategy to $s_i'$ such that $\tc_i(s_i', s_{-i}) < \tc_i(s)$ and $s_i'$ does not contain any closed facility in $s$, then $\Phi(s_i', s_{-i}) < \Phi(s)$.
\end{lemma}

\begin{proof}
	The proof is exactly the same as for Lemma~\ref{lemma-phi-basic}
\end{proof}

\begin{theorem}
	\label{thm-134-fix3-central-multi}
	Suppose there is a central coordinator to determine $\gamma$. When agents deviate, they can drop from at most one facility in each deviation, and are allowed to join as many facilities as they would like to. Then PoS is at most 2 and this bound is tight.
\end{theorem}

\begin{proof}

	This proof is very similar to the proof of Theorem~\ref{thm-134-fix3-central-single}. We also define deviation steps that converge to a stable state, with $\Phi(s)$ decreases in each step of the deviation. For each facility $f_k$, there are two cases:

	\textbf{Case 1, $c(f_k) > \sum_{i|f_k \in s_i} \Q_i(s, f_k)$.} In this case, we close $f_k$ and let each agent $i$ using $f_k$ in $s$ switch its strategy to $\nBR_i(s, f_k)$. Similar to Theorem~\ref{thm-134-fix3-central-single}, we consider the cost of newly disconnected pair of agents and newly connected pair of agents after closing $f_k$. Because the agents are only allowed to drop from one facility, which is $f_k$ in this case, so no agent would have ``unexpected disconnection cost'' when it switches to $\nBR_i(s, f_k)$ from other agents switching away from $\nBR_i(s, f_k)$. The rest of the proof in this case is the same as in the single facility setting, with modified notation.

	\textbf{Case 2, $c(f_k) \le \sum_{i|f_k \in s_i} \Q_i(s, f_k)$.} In this case, we consider pricing strategy $\gamma$, such that $\gamma_i(f_k) = \Q_i(s, f_k)$. First, no agent wants to deviate by only joining some new facilities, but not leaving any facility. This is because every agent is stable without considering facility cost, and their facility cost does not change if they do not leave any facility. By Lemma~\ref{lemma-Q-multi}, for every facility $f_k$, every agent $i$ would not deviate away from $f_k$ with $\gamma_i(f_k) = \Q_i(s, f_k)$. Thus, every agent is stable at ($s$, $\gamma$). The rest of the proof is the same as in the single facility setting.
	
	See the appendix for the full proof.
\end{proof}

\begin{theorem}
	\label{thm-alpha-multi}
	 There always exists an $\alpha$-approximate stable state $(\hat{s}, \gamma)$ such that $\frac{c(\hat{s})}{c(s^*)} \le \frac{2}{\alpha}$.
\end{theorem}

\begin{proof}
	The proof is essentially the same as that of Theorem~\ref{thm-alpha-single}, with the following potential function:
	$$\Phi_{\alpha}(s) = \sum_{i \in \mathcal{A}} \sum_{f_k \in s_i} \cc(i, f_k) + \alpha \sum_{(i,j)| s_i \cap s_j = \emptyset} \dc(i, j) + \sum_{f_k \in \mathcal{F} | f_k \text{ is open}}$$
\end{proof}

\subsection{Agents paying each other}
\label{sec-paid-peering-multi}
The assumptions in this section are the same as in Section~\ref{sec-paid-peering}, but in the multiple facilities setting. We consider the case that agents can pay each other to stabilize the optimal assignment. Formally, for a pair $(i, j)$ with $f_k \in s_i \cap s_j$, $i$ can pay $j$ up to $\dc(i, j)$ to stabilize the current assignment.

In the optimal assignment $s^*$ with a pricing strategy $\gamma$, we will consider the stability of every agent using $f_k$. First, no agent could lower its cost by only joining some other facilities, but not leaving any facility. Suppose to the contrary that there exists an agent $i$ that can deviate to $s_i'$ by only joining some extra facilities to lower its cost. Because $i$ does not leave any facility, then $\rc_i(s)$ decreases means $\tc_i(s)$ also decreases. Any other agent $j$'s $\rc_j(s)$ does not increase, because the connection cost of $j$ does not change, and disconnection cost does not increase. Also, agent $i$ would not join a facility that is not open in $s^*$ because being the only agent at that facility would not benefit it, so the total facility cost is also non-increasing. Thus, in the assignment $(s_i', s^*_{-i})$, the total social cost is lower than that of $s^*$, which is a contradiction.

The above argument satisfies the first condition in Lemma~\ref{lemma-Q-multi}. Thus, every agent $i$ would be stable if $\gamma_i(f_k) - \Delta_i(f_k) \le \Q_i(s^*, f_k)$ for all $f_k \in s^*_i$. For a pair of agents $(i, j)$ using $f_k$ in $s^*$, suppose $\Q_i(s^*, f_k) \ge 0$, and $\Q_j(s^*, f_k) < 0$, which means we can get some payments from $i$ while keeping it stable, but $j$ needs to be paid to keep stable at $s^*$, then we allow $i$ to pay $j$ to stop it from deviating. $i$ would not pay more than $\dc(i, j)$, which is the maximum increase of $i$'s cost as a result of $j$'s deviation. In this section, we consider the stability with payments defined in Section~\ref{sec-pricing-stable-multi}.

 $p_{ij}$ is defined the same as in Section~\ref{sec-paid-peering}, and $\Delta_i(f_k) = \sum_{j|f_k \in s^*_i \cap s^*_j} p_{ji}$.

\begin{theorem}
	\label{paid-peering-thm1-multi}
 If agents can connect to multiple facilities, and we allow agents to pay their neighbors, with $i$ paying $j$ no more than $\dc(i, j)$, then there exist $\gamma$ and payments of players to each other so that the resulting solution $(s^*,\gamma)$ is stable, with $s^*$ being the solution minimizing social cost. In other words, the price of stability becomes 1.
\end{theorem}

\begin{proof}
The proof is similar to Theorem~\ref{paid-peering-thm1-single}, except that each agent is allowed to use multiple facilities and the valid deviations are limited to leaving one facility each time.
See the appendix for the full proof.
\end{proof}

\section{Computation of Optimum Solutions}
In this section, we discuss approximation algorithms to calculate the optimal assignment in polynomial time.

\subsection{Single Facility per Agent}
\label{sec-opt-single}
\begin{theorem}
  If $\forall k, c(f_k) = 0$ and each agent is only allowed to use one facility, then computing the optimum solution $s^*$ is NP-Hard, but there exists a poly-time 2-approximation algorithm.
\end{theorem}

\begin{proof}
	First, notice that this setting is a simple generalization of the multi-way cut problem, which is proved to be NP-Hard \cite{dahlhaus1994complexity}, so this problem is also NP-Hard.

	Next, we will show that this setting can be reduced to the uniform labeling problem. In the uniform labeling problem, we are give an undirected graph and a set of labels. The goal is to assign every node in the graph a label to minimize the sum of two costs: (1) there is a cost $c(x, i)$ for assigning label $x$ to node $i$; (2) there is a separation cost $c(i, j)$ for neighbors $i$ and $j$ with different labels.
	
	Given an instance of our group formation problem, we create a label for each facility, and create a node for each agent. Assign the labeling cost between agent $i$ and facility $f_k$ as the connection cost $\cc(i, f_k)$. Also, for each pair of agents $i, j$ such that $\dc(i, j) > 0$, create an edge between node $i$ and $j$, and assign the separation cost between them as $2\dc(i, j)$. In addition, create a personal label for each agent which corresponds to this agent not joining any facility: this label would have cost 0 for this agent, but a very large cost for all other agents. Using this reduction, it is easy to see that our group formation problem can by solved by converting it to a uniform labeling problem while preserving the approximation factor, and there is a known 2-approximation algorithm for uniform labeling which runs in polynomial time \cite{kleinberg2002approximation}.
\end{proof}

Although there is a polynomial time algorithm that gives a 2-approximation to the optimal solution when the facility costs are 0, the optimal solution is much harder to approximate in the case that facilities costs are not all 0. Consider a weaker setting where all the disconnection costs are 0: then our problem is equivalent to the general facility location problem. \cite{hochbaum1982heuristics} gives a $O(\log n)$-approximation algorithm to this problem and shows that it is harder than the set cover problem, which means it is inapproximable to better than $\Omega(\log n)$ unless P=NP. We will show that in our setting, when the facility costs are arbitrary, and even if the agents are allowed to use multiple facilities, there still exists a polynomial time algorithm that gives a $O(\log n)$-approximation.

\subsection{Multiple Facilities per Agent}
\label{sec-opt-multi}
In this section, we assume each agent is allowed to use multiple facilities, and show that there exist polynomial time algorithms that give a $\min\{O(\log n), m+1\}$-approximation with high probability.


We first model this problem by an integer program and then relax it to a linear program. For every agent $i$ and facility $f_k$, let variable $x_{ik} = 1$ represent that $i$ uses $f_k$, and $x_{ik} = 0$ otherwise. Let $w_{ik} = \cc(i, f_k)$ be the connection cost. For each pair of agents $(i, j)$, let variable $x_{ij} = 1$ represent that $i$ and $j$ are not connected via any facility, and $x_{ij} = 0$ otherwise. Note that $(i,j)$ is still an unordered pair here. Let $w_{ij} = \dc(i, j)$ be the disconnection cost. For each pair of agents $(i, j)$ and facility $f_k$, let variable $x_{ijk} = 1$ represent that $i$ and $j$ are connected via facility $f_k$, and $x_{ijk} = 0$ otherwise. Finally, let $x_k$ represent whether facility $f_k$ is open or not. Then, computing the optimum solution $s^*$ is equivalent to the following integer program:

	\begin{equation}
		\label{opt-IP}
		\begin{array}{ll@{}ll}
			\text{minimize}    & \sum_{i, k} w_{ik} x_{ik} + 2 \sum_{(i, j)} w_{ij} x_{ij} + \sum_{k} c(f_k) x_k & \\
			\text{subject to}  & x_{ijk} \le x_{ik} & \forall (i,j), k \\
												 & x_{ijk} \le x_{jk} & \forall (i,j), k \\
												 & 1 - x_{ij} \le \sum_{k} x_{ijk} & \forall (i,j)\\
												 & x_k \ge x_{ik}  & \forall i, k \\
												 & x_{ijk} \in \{0, 1\} & \forall (i,j), k \\
												 & x_{ij} \in \{0, 1\}  & \forall (i,j)\\
												 & x_k \in \{0, 1\}  & \forall k\\
		\end{array}
	\end{equation}

	Relax the integer program above to a linear program:

	\begin{equation}
		\label{opt-LP}
		\begin{array}{ll@{}ll}
			\text{minimize}    & \sum_{i, k} w_{ik} x_{ik} + 2 \sum_{(i, j)} w_{ij} x_{ij} + \sum_{k} c(f_k) x_k  & \\
			\text{subject to}  & x_{ijk} \le x_{ik} & \forall (i,j), k \\
												 & x_{ijk} \le x_{jk} & \forall (i,j), k \\
												 & 1 - x_{ij} \le \sum_{k} x_{ijk} & \forall (i,j)\\
												 & x_k \ge x_{ik}  & \forall i, k \\
												 & 0 \le x_{ij} \le 1 & \forall (i,j)\\
												 & 0 \le x_{ijk} \le 1 & \forall (i,j), k \\
												 & 0 \le x_k \le 1 & \forall k \\
		\end{array}
		\end{equation}

\begin{algorithm}
\label{opt-k}
Let $x_{ik}^*, x_{ij}^*, x_{k}^*, x_{ijk}^*$ denote the optimal solution to LP \ref{opt-LP}. For all $x_{ik}^* \ge \frac{1}{m+1}$, set $\hat{x}_{ik} = 1$, otherwise $\hat{x}_{ik} = 0$. For all $(i,j)$ and $k$, set $\hat{x}_{ijk} = \min\{\hat{x}_{ik}, \hat{x}_{jk}\}$. Then for all $(i,j)$, set $\hat{x}_{ij} = \max\{0, 1 - \sum_k \hat{x}_{ijk}\} $. Finally, for all $k$, set $\hat{x}_k = \max_i \hat{x}_{ik}$. Return $\hat{x}_{ik}$, $\hat{x}_{ijk}$, $\hat{x}_{ij}$, $\hat{x}_k$ as the solution for IP \ref{opt-IP}.
\end{algorithm}

\begin{theorem}
  Algorithm \ref{opt-k} gives a $(m+1)$-approximation to the optimal solution of IP \ref{opt-IP}.
\end{theorem}

\begin{proof}

	Algorithm \ref{opt-k} gives a valid solution to IP~\ref{opt-IP}. First, all variables are either set to 0 or 1. We first round all $\hat{x}_{ik}$, then set $\hat{x}_{ijk} = \min\{\hat{x}_{ik}, \hat{x}_{jk}\}$, so $\hat{x}_{ijk} \le \hat{x}_{ik}$ and $\hat{x}_{ijk} \le \hat{x}_{jk}$. $\hat{x}_{ij} = \max\{0, 1 - \sum_k \hat{x}_{ijk}\}$ guarantees $1 - \hat{x}_{ij} \le \sum_{k} \hat{x}_{ijk}$. Finally $\hat{x}_k = \max_i \hat{x}_{ik}$, so $\hat{x}_k \ge \hat{x}_{ik}$ is satisfied.

	By the rounding in Algorithm \ref{opt-k}, for all $i$ and $k$, $\hat{x}_{ik} \le (m+1) x_{ik}^*$. We will then compare $x_{ij}^*$ with $\hat{x}_{ij}$ and $x_{k}^*$ with $\hat{x}_{k}$.

	$\hat{x}_{ij} = 1$ if and only if $\sum_k \hat{x}_{ijk} = 0$, which means $\forall k$, $\min\{\hat{x}_{ik}, \hat{x}_{jk}\} = 0$, then it must be the case that $\min\{x_{ik}^*, x_{jk}^*\} < \frac{1}{m+1}$. Because $x_{ijk}^* \le \min\{x_{ik}^*, x_{jk}^*\}$ and $1 - x_{ij}^* \le \sum_{k} x_{ijk}^*$, we know $x_{ij}^* \ge 1 - m \times \frac{1}{m+1} = \frac{1}{m+1}$. So $\hat{x}_{ij} = 1$ only if $x_{ij}^* \ge \frac{1}{m+1}$, and thus $\hat{x}_{ij} \le (m+1) x_{ij}^*$.

	$\hat{x}_{k} = 1$ only if $\exists i$, $\hat{x}_{ik} = 1$, which means $\exists i$, $x_{ik}^* \ge \frac{1}{m+1}$. So $x_{k}^* = \max_{i} x_{ik}^* \ge \frac{1}{m+1}$, then $\hat{x}_{k} \le (m+1) x_{k}^*$.

	Because all the variables given by the rounding method in Algorithm \ref{opt-k} are at most $(m+1)$ times the optimal LP solution, we have that Algorithm \ref{opt-k} gives a $(m+1)$-approximation to the optimal solution of LP \ref{opt-LP}, and \ref{opt-IP}.
\end{proof}

\begin{algorithm}
	\label{opt-logn}
	Let $x_{ik}^*, x_{ij}^*, x_{k}^*, x_{ijk}^*$ denote the optimal solution to LP \ref{opt-LP}. For each facility $f_k$, we apply correlated randomized rounding on all $x_{ik}^*$ as follows: first order all agents $i$ by increasing order of $x_{ik}^*$. Without loss of generality, suppose $x_{1k}^* \le x_{2k}^* \le \dots \le x_{nk}^*$. With probability $x_{1k}^*$, assign $x_{ik} = 1$ for all $i$. With probability $x_{2k}^* - x_{1k}^*$, assign $x_{ik} = 1$ for all $i \ge 2$, and $x_{ik} = 0$ for all $i < 2$. With probability $x_{jk}^* - x_{(j-1)k}^*$, assign $x_{ik} = 1$ for all $i \ge j$, and $x_{ik} = 0$ for all $i < j$. Finally, with probability $1 - x_{nk}$, assign  $x_{ik} = 0$ for all $i$. Once all $x_{ik}$ are assigned, for all $(i,j)$ and $k$, set $x_{ijk} = \min\{x_{ik}, x_{jk}\}$. Then for all $(i,j)$, set $x_{ij} = \max\{0, 1 - \sum_k x_{ijk}\} $. For all $k$, set $x_k = \max_i x_{ik}$. Repeat this randomized rounding process for $4 \ln 10n$ times, then assign $\hat{x}_{ik} = 1$ if $x_{ik}$ is assigned to 1 in any one of the $4 \ln 10n$ runs, otherwise $\hat{x}_{ik}=0$. Assign $\hat{x}_{ijk} = \min\{\hat{x}_{ik}, \hat{x}_{jk}\}$. Assign $\hat{x}_{ij} = 1$ if and only if $x_{ij}$ is assigned to 1 in every single run. Finally, assign $\hat{x}_{k} = 1$ if $x_{k}$ is assigned to 1 in any single run, otherwise $\hat{x}_{k} = 0$. Return $\hat{x}_{ik}$, $\hat{x}_{ijk}$, $\hat{x}_{ij}$, $\hat{x}_k$ as the solution for IP \ref{opt-IP}.
\end{algorithm}

\begin{theorem}
   With high probability, the social cost of the solution given by Algorithm \ref{opt-logn} is no more than $O(\ln n)\cdot \rc(s^*)$.
\end{theorem}

\begin{proof}

	First we show that Algorithm \ref{opt-logn} gives a valid solution to IP~\ref{opt-IP}. All variables are either set to 0 or 1. $\hat{x}_{ik}=1$ if $x_{ik}$ is assigned to 1 in any single run, and we set $\hat{x}_{ijk} = \min\{\hat{x}_{ik}, \hat{x}_{jk}\}$, so $\hat{x}_{ijk} \le \hat{x}_{ik}$ and $\hat{x}_{ijk} \le \hat{x}_{jk}$. $\hat{x}_{ij} = 1$ if and only if $x_{ij}$ is assigned to 1 in {\em every} single run, so $1 - \hat{x}_{ij} \le \sum_{k} \hat{x}_{ijk}$ holds, because when $x_{ij}$ is assigned to 1 in every run, it means $x_{ijk} = 0$ for all $k$ in every run, then it must be $x_{ik} = 0$ and $x_{jk} = 0$ in every run, so $\hat{x}_{ik} = \hat{x}_{jk} = \hat{x}_{ijk} = 0$. When $\hat{x}_{ij}$ is assigned to 1, $1 - \hat{x}_{ij} \le \sum_{k} \hat{x}_{ijk}$ obviously always holds. Finally, $\hat{x}_{k} = 1$ if $x_{k} = 1$ in any single run, so $\hat{x}_k \ge \hat{x}_{ik}$ must hold because $\hat{x}_{ik}$ is also assigned to 1 if $x_{ik} = 1$ in any single run.

	We will first show that in one run of the randomized rounding method in Algorithm \ref{opt-logn}, $P[x_{ik} = 1] = x_{ik}^*$ and $P[x_{k} = 1] = x_{k}^*$. For any agent $i$, the probability that $x_{ik}$ is set to 1 is exactly $x_{1k}^* + (x_{2k}^* - x_{1k}^*) + \dots + (x_{ik}^* - x_{(i-1)k}^*) = x_{ik}^*$.  $x_{k} = 1$ if and only if $\exists i$, $x_{ik} = 1$, so $P[x_{k} = 1] = \max_{i}\{x_{ik}^*\} = x_{k}^*$.

	$x_{ijk} = 1$ if and only if $\min\{x_{ik}, x_{jk}\} = 1$. By our randomized rounding strategy, $P[\min\{x_{ik}, x_{jk}\} = 1] = \min\{x_{ik}^*, x_{jk}^*\}$. (Note: this is the key point which makes our correlated randomized rounding strategy work. The rounding of $x_{ik}$ and $x_{jk}$ is {\em not} independent, so it can never happen that lower $x_{ik}^*$ gets rounded to 1 but higher one does not.) Thus, $P[x_{ijk} = 1] = \min\{x_{ik}^*, x_{jk}^*\}$. Because $x_{ijk}^* \le \min\{x_{ik}^*, x_{jk}^*\}$, $P[x_{ijk} = 1] \ge x_{ijk}^*$.

	$x_{ij} = 1$ if and only if $\sum_k x_{ijk} = 0$, so in any single run of the randomized algorithm:
	\begin{align*}
		P[x_{ij} = 1] &= P[x_{ij1}=0] \times P[x_{ij1}=0] \times \dots \times P[x_{ijm}=0] \\
					&\le \prod_{k = 1}^m (1 - x_{ijk}^*) \\
					&\le \left( \frac{\sum_{k=1}^m (1-x_{ijk}^*)}{m} \right) ^ m \\
					&= \left( \frac{m - \sum_{k=1}^m x_{ijk}^*}{m} \right) ^ m \\
	\end{align*}

	Because $1 - x_{ij}^* \le \sum_{k} x_{ijk}^*$:

	\begin{equation}
		\label{opt-logn-eq1}
		P[x_{ij} = 1] \le \left( \frac{m - (1 - x_{ij}^*)}{m} \right) ^ m
	\end{equation}

	When $x_{ij}^* \le \frac{1}{2}$,
	$$P[x_{ij} = 1] \le \left( \frac{m - \frac{1}{2}}{m} \right) ^ m \le \frac{1}{\sqrt{e}}$$

	When we repeat the randomized rounding process for $t = 4 \ln 10n$ times, if $x_{ij}^* \le \frac{1}{2}$, the probability that $x_{ij}$ is set to 1 in every single run is at most $\frac{1}{\sqrt{e}^t} = \frac{1}{10 n^2}$. Since there are $n(n-1)$ pairs of agents $(i,j)$, then the probability that any $x_{ij}^* \le \frac{1}{2}$ and $\hat{x}_{ij} = 1$ is at most $n^2 \frac{1}{10 n^2} = 0.1$.
	

	When $x_{ij}^* > \frac{1}{2}$, we will prove that $P[x_{ij} = 1]$ can be bounded by a linear function $a x_{ij}^*$ for some constant $a$. We have shown $P[x_{ij} = 1] \le \frac{1}{\sqrt{e}}$ when $x_{ij}^* = \frac{1}{2}$. By Inequality~\ref{opt-logn-eq1}, when $x_{ij}^* = 1$, $P[x_{ij} = 1] \le 1$. As a function of $x_{ij}^*$, the right-hand side of Inequality~\ref{opt-logn-eq1} is strictly increasing and convex.
%
%
%
%
%
Therefore, this formula can be upper bounded by a linear function passing through the points $(\frac{1}{2}, \frac{1}{\sqrt{e}})$ and $(1, 1)$, which means:

	$$P[x_{ij} = 1] \le 2(1 - \frac{1}{\sqrt{e}}) x_{ij}^* + \frac{2}{\sqrt{e}} - 1$$

	When $x_{ij}^* \in [\frac{1}{2}, 1]$, the right hand side of the above inequality is at most $\frac{2}{\sqrt{e}} x_{ij}^*$, so:
	
	$$P[x_{ij} = 1 |  x_{ij}^*>\frac{1}{2}] \le \frac{2}{\sqrt{e}} x_{ij}^*$$

	Let $A$ denote the event that there exist a pair $(i, j)$ such that $x_{ij}^* \le \frac{1}{2}$ and $\hat{x}_{ij} = 1$ ($x_{ij}$ set to 1 in every single run). By our analysis earlier, $P[A] \le 0.1$. Let $\hat{s}$ denote the assignment returned by Algorithm \ref{opt-logn}. We will consider $\rc(\hat{s})$ with condition $A$ and $\bar{A}$ separately. First calculate the expected social cost of $\hat{s}$ with condition $\bar{A}$:
	\begin{align}
		E[\rc(\hat{s}) | \bar{A}]
		&=E[\sum_{i, k} w_{ik} \hat{x}_{ik} + 2 \sum_{(i, j)} w_{ij} \hat{x}_{ij} + \sum_{k} c(f_k) \hat{x}_k | \bar{A}] \nonumber \\
		&= \sum_{i, k} w_{ik} P[\hat{x}_{ik}=1 | \bar{A}] + 2 \sum_{(i, j)} w_{ij} P[\hat{x}_{ij}=1 | \bar{A}] + \sum_{k} c(f_k) P[\hat{x}_k=1 | \bar{A}] \label{opt-logn-eq2}
	\end{align}

	In any single run, it is always true that $P[x_{ik} = 1] = x_{ik}^*$ and $P[x_{k} = 1] = x_{k}^*$, with any condition. Remember $\hat{x}_{ik} = 1$ if $x_{ik} = 1$ in any one of the $4 \ln 10n$ runs, so $P[\hat{x}_{ik} = 1 | \bar{A}] \le x_{ik}^*4 \ln 10n $. By the same reason, $P[\hat{x}_k=1 | \bar{A}] \le x_k^*4 \ln10n $. Finally, we bound $P[\hat{x}_{ij}=1 | \bar{A}]$ by further decomposing $\bar{A}$ into two cases: $x_{ij}^* \le \frac{1}{2}$ and $x_{ij}^* > \frac{1}{2}$. When $x_{ij}^* \le \frac{1}{2}$, by the definition of event $\bar{A}$, for every pair of agents $(i, j)$ that $x_{ij}^* \le \frac{1}{2}$, we have $\hat{x}_{ij}=0$. Thus, $P[\hat{x}_{ij}=1 | \bar{A}] = P[\hat{x}_{ij}=1 | x_{ij}^* > \frac{1}{2}] \le \frac{2}{\sqrt{e}} x_{ij}^*$. Apply these bounds to Inequality \ref{opt-logn-eq2}:
	\begin{align*}
		&E[\rc(\hat{s}) | \bar{A}] \\
		&\le 4 \ln 10n \sum_{i, k} w_{ik} x_{ik}^* + \frac{4}{\sqrt{e}} \sum_{(i, j)} w_{ij} x_{ij}^* + 4 \ln 10n  \sum_{k} c(f_k) x_k^*\\
		&\le (4\ln 10n) (\sum_{i, k} w_{ik} x_{ik}^* + 2 \sum_{(i, j)} w_{ij} x_{ij}^* + \sum_{k} c(f_k) x_k^*)\\
		&= (4\ln 10n) \rc(s^*)
	\end{align*}

	By Markov's inequality:

	$$P[\rc(\hat{s}) \ge (20\ln 10n) \rc(s^*) | \bar{A}] \le P[\rc(\hat{s}) \ge 5 E[\rc(\hat{s}) | \bar{A}]] \le \frac{1}{5}$$

	Finally, we bound the probability that $\rc(\hat{s}) \ge (20\ln 10n) \rc(s^*)$ by considering event $A$ and $\bar{A}$ separately:
	\begin{align*}
		&P[\rc(\hat{s}) \ge (20 \ln 10n) \rc(s^*)] \\
		&= P[\rc(\hat{s}) \ge (20 \ln 10n) \rc(s^*) | A] \times P[A] + P[\rc(\hat{s}) \ge (20 \ln 10n) \rc(s^*) | \bar{A}] \times P[\bar{A}] \\
		&\le 1 \times 0.1 + 0.2 \times 1\\
		&= 0.3
	\end{align*}
Thus, we have shown that with probability of at least $0.7$, the social cost of the solution given by Algorithm \ref{opt-logn} is no more than $(20 \ln 10n) \rc(s^*)$.
\end{proof}

\subsection*{Acknowledgements} This work was partially supported by NSF awards CNS-1816396 and CCF-1527497.



\appendix

\section{Proofs in Section \ref{sec-multi}}

The proofs for agents being allowed to connect to multiple facilities are almost the same as for the setting where they can only connect to a single facility. We include the proofs of these results in this appendix for completeness, and only explain the difference between the single and multiple facility setting in the main body of this paper.

\subsection{Proof of Theorem \ref{thm-14-multi}}

\begin{proof}
	With the above definition of $\tPhi(s)$, when an agent $i$ switches its strategy from $s_i$ to $s_i'$, it is easy to see that the change of $i$'s cost is captured exactly by the change of $\tPhi(s)$:

	$$\tc_i(s_i', s_{-i}) - \tc_i(s) = \tPhi(s_i', s_{-i}) - \tPhi(s)$$

	Thus, $\tPhi(s)$ is an exact potential function.

	The total social cost in this case is:
	$$c(s) = \sum_{i \in \mathcal{A}} \sum_{f_k \in s_i} \cc(i, f_k) + 2 \sum_{(i, j) | s_i \cap s_j= \emptyset} \dc(i, j)$$

	Consider the assignment $\hat{s}$ that minimizes $\tPhi$. $\hat{s}$ must be stable, because the exact potential function $\tPhi$ is minimized, so no agent could deviate to lower its cost. Bound the social cost of $\hat{s}$:

	\begin{align*}
	\rc(\hat{s}) &= \sum_{i \in \mathcal{A}} \sum_{f_k \in \hat{s}_i} \cc(i, f_k) + 2 \sum_{(i, j) | \hat{s}_i \cap \hat{s}_j = \emptyset } \dc(i, j)\\
	&\le 2 \tPhi(\hat{s})\\
	&< 2 \tPhi(s^*)\\
	&\le 2 \rc(s^*)
	\end{align*}

	Denote the final stable state as $\hat{s}$. Similar to the analysis in Theorem \ref{thm-14-single},

	$$\rc(\hat{s}) \le 2 \Phi(\hat{s}) <  2 \Phi(s^*) \le 2 \rc(s^*)$$

	Thus, PoS is at most 2.

	Consider an example with one facility and two agents. $\cc(1, f_1) = 0$, $\cc(2, f_1) = 1-\epsilon$, $\dc(1,2) = 1$. The only stable state is agent 1 and 2 both do not use $f_k$, i.e., $s_i \ne s_j$. When $\epsilon$ approaches 0, PoS approaches 2.
\end{proof}

\subsection{Proof of Theorem \ref{thm-134-fix3-central-multi}}

\begin{proof}
	We define a deviation process that converges to a stable state, with $\Phi(s)$ decreases in each step of the deviation.

	Start with the optimal assignment $s^*$, if there exists an agent $i$ that when $i$ switches its strategy to $s_i'$, in which $s_i'$ does not contain any closed facility in $s^*$ and $\tc_i(s_i', s^*_{-i}) < \tc_i(s^*)$, then let agent $i$ switch to $s_i'$. Select another agent to repeat this process until no such agent exists. By Lemma~\ref{lemma-phi-basic-multi}, $\Phi(s)$ decreases during each step in this process.

	We know that no agent $i$ can decrease $\tc_i(s)$ by switching to another valid deviation that does not contain any closed facility in $s$. Note that even if $i$ switches to a set of facilities that contains a closed facility in $s$, it would not be able to lower $\tc_i(s)$. Suppose there exists a valid deviation $s_i'$ that $\tc_i(s_i', s_{-i}) < \tc_i(s)$, then we can always create another deviation $s_i''$ by removing all facilities not open in $s$ from $s_i'$, and it is obvious that $\tc_i(s_i'', s_{-i}) \le \tc_i(s_i', s_{-i})$. This contradicts the fact that no agent $i$ can decrease $\tc_i(s)$ by switching to any valid deviation that does not contain any closed facility in $s$.	Thus, in this ``stable'' state, every agent is stable if it is charged 0 facility cost.
	
	By Lemma~\ref{lemma-Q-multi}, for any agent $i$ that uses $f_k$, we can charge a facility cost of $\Q_i(s, f_k)$ to agent $i$ while keeping it stable:

	\begin{align}
		\Q_i(s, f_k)
		&= \tc_i(\nBR_i(s, f_k), s_{-i}) - \tc_i(s) \nonumber \\
		&= \sum_{f_k' \in \nBR_i(s, f_k) - s_i} \cc(i, f_k') + \sum_{j | f_k \in s_j, \nBR_i(s, f_k) \cap s_j = \emptyset} \dc(i,  j) - \cc(i, f_k) - \sum_{j | (\nBR_i(s, f_k) - s_i) \cap s_j \ne \emptyset} \dc(i,  j)\label{thm-134-fix3-central-multi-eq1}
	\end{align}

	Remember $\nBR_i(s, f_k)$ denote $i$'s best response given $i$ is forced to stop using $f_k$ (and forbidden to join $f_k$ again), while the assignment of all other agents do not change. In the multiple facility setting, we assume each agent can leave at most one facility in every deviation. So if we compare $\nBR_i(s, f_k)$ with $s_i$, $f_k$ is the only facility that is in $s_i$ but not in $\nBR_i(s, f_k)$. For each facility $f_k$, consider the following two cases:\\

	\textbf{Case 1, $c(f_k) > \sum_{i|f_k \in s_i} \Q_i(s, f_k)$.} In this case, we close $f_k$ and let each agent $i$ using $f_k$ in $s$ deviate to $\nBR_i(s, f_k)$. Denote the assignment after closing $f_k$ as $s'$, then:

	\begin{align*}
	\Phi(s') - \Phi(s)
	&= -c(f_k) + \sum_{i|f_k \in s_i} (\sum_{f_k' \in \nBR_i(s, f_k) - s_i} \cc(i, f_k') - \cc(i, f_k)) \\
	&+ \sum_{(i,j) | f_k \in s_i \cap s_j, s'_i \cap s'_j = \emptyset} \dc(i, j) - \sum_{(i,j) | s_i \cap s_j = \emptyset, s'_i \cap s'_j \ne \emptyset} \dc(i, j)\\
	\end{align*}

	Note that only agents using $f_k$ in $s$ change their strategies, and all other agents keep their strategies at $s$. Thus, the newly disconnected agent pairs in $s'$ are at most those only share $f_k$ in $s$:

	\begin{align*}
	\sum_{(i,j) | f_k \in s_i \cap s_j, s'_i \cap s'_j = \emptyset} \dc(i, j)
	&\le \sum_{i | f_k \in s_i} \sum_{j | f_k \in s_j, s_i \cap s_j = f_k } \dc(i,  j)\\
	&\le \sum_{i | f_k \in s_i} \sum_{j | f_k \in s_j, \nBR_i(s, f_k) \cap s_j = \emptyset} \dc(i,  j)
	\end{align*}
	
	The last line follows because by the definition of $\nBR_i(s, f_k)$, $f_k \notin \nBR_i(s, f_k)$. Therefore, if $s_i \cap s_j = f_k$, then it must be $\nBR_i(s, f_k) \cap s_j = \emptyset$.

	Also, the newly connected agent pairs in $s'$ are created by the agents using $f_k$ deviating to their $\nBR_i(s, f_k)$, i.e.,

	$$ \sum_{(i,j) |  s_i \cap s_j = \emptyset, s'_i \cap s'_j \ne \emptyset} \dc(i, j) \ge  \sum_{i| f_k \in s_i} \sum_{j | (\nBR_i(s, f_k) - s_i) \cap s_j \ne \emptyset} \dc(i,  j) $$

  With the condition of \textbf{Case 1}, we can bound $\Phi(s') - \Phi(s)$ by:

	\begin{align*}
	&\Phi(s') - \Phi(s)\\
	&= -c(f_k) + \sum_{i | f_k \in s_i} (\sum_{f_k' \in \nBR_i(s, f_k) - s_i} \cc(i, f_k') - \cc(i, f_k))
	+ \sum_{(i,j) | f_k \in s_i \cap s_j, s'_i \cap s'_j = \emptyset} \dc(i, j) - \sum_{(i,j) |  s_i \cap s_j = \emptyset, s'_i \cap s'_j \ne \emptyset} \dc(i, j)\\
	&\le -c(f_k) + \sum_{i | f_k \in s_i} (\sum_{f_k' \in \nBR_i(s, f_k) - s_i} \cc(i, f_k') - \cc(i, f_k))
	+ \sum_{ i | f_k \in s_i} \sum_{j | f_k \in s_j, \nBR_i(s, f_k) \cap s_j = \emptyset} \dc(i, j) \\
	&- \sum_{i | f_k \in s_i} \sum_{j | (\nBR_i(s, f_k) - s_i) \cap s_j \ne \emptyset} \dc(i,  j)\\
	&= -c(f_k) + \sum_{i | f_k \in s_i} (\sum_{f_k' \in \nBR_i(s, f_k) - s_i} \cc(i, f_k') + \sum_{j | f_k \in s_j, \nBR_i(s, f_k) \cap s_j = \emptyset} \dc(i, j)
	- \cc(i, f_k) - \sum_{j | (\nBR_i(s, f_k) - s_i) \cap s_j \ne \emptyset} \dc(i,  j))\\
	&= -c(f_k) + \sum_{i | f_k \in s_i} \Q_i(s, f_k)\\
	&< 0 \\
  \end{align*}

	Thus, $\Phi(s)$ does not increase after $f_k$ is closed and the agents deviate to $\nBR_i(s, f_k)$. \\

	Then repeat the above two steps: let agents switch strategies one at a time to reach a 'stable' state $s$ if all agents ignore the facility cost, then if there exist a facility $f_k$ that satisfies the condition in \textbf{Case 1}, we close $f_k$ and let every agent $i$ using it switches its strategy to $\nBR_i(s, f_k)$. We repeat these two steps until we reach state $s$ that every agent is ``stable'' if all agents ignore the facility cost, and every open facility does not satisfy \textbf{Case 1}. Note that $\Phi(s)$ decreases in each step, so this process always converges to such an assignment $s$. Then each open facility must satisfy the following \textbf{Case 2}:

	\textbf{Case 2, $c(f_k) \le \sum_{i|f_k \in s_i} \Q_i(s, f_k)$.} In this case, we consider pricing strategy $\gamma$, that $\gamma_i(f_k) = \Q_i(s, f_k)$. First, no agent wants to deviate by only joining some new facilities, but not leaving any facility. Because every agent is stable without considering facility cost, and their facility cost does not change if they do not leave any facility. By Lemma~\ref{lemma-Q-multi}, for every facility $f_k$, every agent $i$ would not deviate away from $f_k$ with $\gamma_i(f_k) = \Q_i(s, f_k)$. Thus, every agent is stable at ($s$, $\gamma$). Also, by definition, $\Q_i(s, f_k) = \tc_i(\nBR_i(s, f_k), s_{-i}) - \tc_i(s)$. Because every agent is stable without considering facility costs, so $\tc_i(s) \le \tc_i(\nBR_i(s, f_k), s_{-i})$. Therefore, we are charging a non-negative facility cost to each agent. By the condition of \textbf{Case 2}, $\sum_{i|f_k \in s_i} \gamma_i(f_k) \ge c(f_k)$. If $\sum_{i|f_k \in s_i} \gamma_i(f_k) > c(f_k)$, then to get a budget-balanced cost assignment, we can lower the facility cost of some agents to make the sum of $\gamma_i(f_k)$ exactly $c(f_k)$, because the agents would not deviate with $\gamma_i(f_k)$, then they would not deviate with a lower facility cost. Thus, we have reached a stable state.

  $\Phi(s)$ decreases in each deviation, so it is a ordinal potential function for our deviation processes.

	The total social cost is:

	$$c(s) = \sum_{f_k \in \mathcal{F}| f_k \text{ is open}} c(f_k) + \sum_{i \in \mathcal{A}} \sum_{f_k \in s_i} \cc(i, f_k) + 2\sum_{(i, j)| s_i \cap s_j = \emptyset} \dc(i, j)$$

	Denote the final stable state as $\hat{s}$. Similar to the analysis in Theorem \ref{thm-14-single},

	$$\rc(\hat{s}) \le 2 \Phi(\hat{s}) <  2 \Phi(s^*) \le 2 \rc(s^*)$$

	Thus, PoS is at most 2.\\

	See Theorem~\ref{thm-14-single} for the lower bound example (assuming $c(f_k) = 0$).
\end{proof}

\subsection{Proof of Theorem\ref{paid-peering-thm1-multi}}

\begin{proof}
	We construct a circulation network \cite{kleinberg2006algorithm} for each facility $f_k$ as follows: start from the optimal assignment $s^*$. Create a node for each agent $i$ such that $f_k \in s^*_i$, and we set a supply of $\Q_i(s^*, f_k)$ to it. (note that this value might be negative, in which case the node has a demand instead of supply). For each pair of agents $(i, j)$, we create directed edges from $i$ to $j$ and from $j$ to $i$, both with capacity $\dc(i, j)$. Create a node $f_k$ with a supply of $-c(f_k)$ and and edge from each node to $f_k$ with infinite capacity. Finally, create a dummy node $z$ with a demand of the sum of supplies of all other nodes, and add an edge from $f_k$ to $z$ with infinite capacity. This is to make sure the total supply meets the total demand in the network.

	Suppose there is a feasible solution, then we can get the flow on each edge to create a stable state: First, denote the flow from any node $i$ to $j$ as $v_{ij}$. For each pair of nodes $i$ and $j$ such that $f_k \in s^*_i = s^*_j$, set $p_{ij} = v_{ij} - v_{ji}$ and $p_{ji} = v_{ji} - v_{ij}$. Also, for every agent $i$ such that $f_k \in s^*_i$, set $\gamma_i(f_k) = v_{if_k}$. A feasible solution guarantees that facility $f_k$ is fully paid, because $\sum_{i|f_k \in s^*_i} \gamma_i(f_k) = \sum_{i|f_k \in s^*_i} v_{if_k} = c(f_k)$. Also, every agent is stable. First, by the definition of $\Delta_i(f_k)$ in this section, $\Delta_i(f_k) = \sum_{j|f_k \in s^*_i \cap s^*_j} p_{ji} = \sum_{j|f_k \in s^*_i \cap s^*_j} (v_{ji} - v_{ij})$. For agent $i$ such that $f_k \in s^*_i$, the supply of node $i$ is $\Q_i(s^*, f_k)$, which equals the total flow going out of $i$ minus the total flow going into $i$:

	$$\Q_i(s^*, f_k) = v_{if_k} + \sum_{j|f_k \in s^*_i \cap s^*_j} (v_{ij} - v_{ji}) = \gamma_i(f_k) - \Delta_i(f_k)$$
	
	By Lemma~\ref{lemma-Q-multi}, every agent is stable with because $\forall f_k$ and $i$, $\gamma_i(f_k) - \Delta_i(f_k) \le \Q_i(s^*, f_k)$.

	This circulation network is feasible if and only if we can stabilize $s^*$ by allowing agents to pay their neighbors. Also, the facility is fully paid for (budget balanced). By a standard Max-Flow and Min-Cut analysis \cite{kleinberg2006algorithm}, if for any subset of nodes in the circulation network, the total supply of the subset plus the total capacity of edges going into the subset is non-negative, then the circulation network is feasible.
	
		First consider any subset that includes $z$. If the subset does not include $f_k$, then there must be an edge with infinite capacity going into the subset, so the conclusion holds.
	
		Next, consider a subset includes $f_k$ and $z$. If the subset does not include all agents in $f_k$, then there must be an edge with infinite capacity going into the subset, so the conclusion holds. If the subset does include all agents in $f_k$, then by the definition of $z$, the total supply is 0.
	
		Then, consider a subset includes $f_k$ but not $z$. If the subset does not include all agents in $f_k$, then the conclusion still holds. If the subset is actually all the nodes in the network, then the total supply is:
	
		$$ \sum_{i \in \mathcal{A}} \Q_i(s^*, f_k) - c(f_k)$$
	
		Suppose the total supply is negative instead. Then consider an assignment $s'$ that $f_k$ is closed, and every agent $i$ that uses $f_k$ in $s^*$ switch its strategy to $\nBR_i(s^*, f_k)$. For any agent $j$ that does not use $f_k$ in $s^*$, $j$ stay at $s^*_j$. It it easy to see that $\tc_j(s') \le \tc_j(s^*)$ for every $j$ that $f_k \notin s^*_j$. For any agent $i$ that $f_k \in s^*_i$, it must be $\tc_i(s') \le \tc_i(\nBR_i(s^*, f_k), s^*_{-i})$. This is because $\tc_i(\nBR_i(s^*, f_k), s^*_{-i})$ is the cost that only $i$ switches to $\nBR_i(s^*, f_k)$ with all other agents stay at $s^*$, while in $s'$ only agents using $f_k$ switches their strategies. Because agents using $\nBR_i(s^*, f_k)$ in $s^*$ all stay at $s^*$, $i$ would not get any ``unexpected cost'' in $s'$, so $\tc_i(s') \le \tc_i(\nBR_i(s^*, f_k), s^*_{-i})$. Also, because $f_k$ is closed in $s'$, the facility cost decreases by at least $c(f_k)$. No new facility will open in $s'$ (compared to $s^*$) because we have excluded this possibility in the definition of $\nBR_i(s^*, f_k)$. Thus, the total social cost of $s'$ increases by at most:
	
		$$-c(f_k) + \sum_{i | f_k \in s^*_i} (\tc_i(\nBR_i(s^*, f_k), s^*_{-i}) - \tc_i(s^*))
		= -c(f_k) + \sum_{i | f_k \in s^*_i} \Q_i(s^*, f_k)$$
	
		By our assumption, this number is negative, which means $s'$ has less total cost than $s^*$, which contradicts the fact that $s^*$ is optimal. Thus, the total supply must be non-negative in this case.
	
		We denote $\nBR_i(s^*, f_k)$ as $R_i$ for simplification in the following proof. Now consider a subset of nodes does not include $f_k$. Suppose there exists a subset $\mathcal{B}$, that the total supply of nodes in $\mathcal{B}$ plus the total capacity of edges going into the subset is negative, i.e.,
	
		\begin{equation}
			\label{paid-peering-multi-eq1}
			\sum_{i \in \mathcal{B}} \Q_i(s^*, f_k) + \sum_{(i, j) | i \in \mathcal{B}, j \notin \mathcal{B}, f_k \in s^*_j} \dc(i, j) < 0
		\end{equation}
	
		By the definition of $\Q_i(s^*, f_k)$:
		\begin{align}
			&\sum_{i \in \mathcal{B}} \Q_i(s^*, f_k) \nonumber \\
			&= \sum_{i \in \mathcal{B}} (\sum_{f_k' \in R_i - s^*_i} \cc(i, f_k') + \sum_{j | f_k \in s^*_j, R_i \cap s^*_j = \emptyset} \dc(i,  j) - \cc(i, f_k) - \sum_{j | (R_i - s^*_i) \cap s^*_j \ne \emptyset} \dc(i,  j)) \nonumber \\
			&= \sum_{i \in \mathcal{B}} \sum_{f_k' \in R_i - s^*_i} \cc(i, f_k') + \sum_{i \in \mathcal{B}} \sum_{j | f_k \in s^*_j, R_i \cap s^*_j = \emptyset} \dc(i,  j) - \sum_{i \in \mathcal{B}} \cc(i, f_k) - \sum_{i \in \mathcal{B}} \sum_{j | (R_i - s^*_i) \cap s^*_j \ne \emptyset} \dc(i,  j) \label{paid-peering-multi-eq2}
		 \end{align}
	
			Decompose $\sum_{i \in \mathcal{B}} \sum_{j | f_k \in s^*_j, R_i \cap s^*_j = \emptyset} \dc(i,  j)$ into two parts:
	
			\begin{align*}
				\sum_{i \in \mathcal{B}} \sum_{j | f_k \in s^*_j, R_i \cap s^*_j = \emptyset} \dc(i,  j)
				&= \sum_{i \in \mathcal{B}} (\sum_{j \in \mathcal{B},f_k \in s^*_j, R_i \cap s^*_j = \emptyset} \dc(i, j) +  \sum_{j \notin \mathcal{B}, f_k \in s^*_j, R_i \cap s^*_j = \emptyset} \dc(i, j) ) \\
				&= \sum_{i \in \mathcal{B}} \sum_{j \in \mathcal{B}, R_i \cap s^*_j = \emptyset} \dc(i, j) + \sum_{i \in \mathcal{B}} \sum_{j \notin \mathcal{B}, f_k \in s^*_j, R_i \cap s^*_j = \emptyset} \dc(i, j) \\
				&= \sum_{i \in \mathcal{B}} \sum_{j \in \mathcal{B}, R_i \cap s^*_j = \emptyset}  \dc(i, j) + \sum_{(i, j) | i \in \mathcal{B}, j \notin \mathcal{B},  f_k \in s^*_j, R_i \cap s^*_j = \emptyset} \dc(i, j) \\
				&\ge 2 \sum_{(i, j) | i \in \mathcal{B}, j \in \mathcal{B}, R_i \cap R_j = \emptyset}  \dc(i, j) + \sum_{(i, j) | i \in \mathcal{B}, j \notin \mathcal{B},  f_k \in s^*_j, R_i \cap s^*_j = \emptyset} \dc(i, j) \\
			\end{align*}
	
			Remember we only create nodes for every agent $j$ that $s^*_j = f_k$, and $\mathcal{B}$ is a subset of the nodes, so $\{j \in \mathcal{B}, f_k \in s^*_j, R_i \cap s^*_j = \emptyset\} = \{j \in \mathcal{B}, R_i \cap s^*_j = \emptyset\}$ in the first line of the inequality above. The last line of the inequality above follows because by the definition of $R_j$, because $f_k \notin R_i$, and an agent can only leave one facility when it deviates, so $R_j$ includes all facilities in $s^*$ except $f_k$. Thus, $\{R_i \cap R_j = \emptyset\} \subseteq \{R_i \cap s^*_j = \emptyset \}$.
	
			Together with Inequality~\ref{paid-peering-multi-eq1} and ~\ref{paid-peering-multi-eq2},
			\begin{align}
				&\sum_{i \in \mathcal{B}} \sum_{f_k' \in R_i - s^*_i} \cc(i, f_k')
				+ 2 \sum_{(i, j) | i \in \mathcal{B}, j \in \mathcal{B}, R_i \cap R_j = \emptyset}  \dc(i, j)
				+ \sum_{(i, j) | i \in \mathcal{B}, j \notin \mathcal{B}, f_k \in s^*_j, R_i \cap s^*_j = \emptyset} \dc(i, j) \nonumber \\
				&- \sum_{i \in \mathcal{B}} \cc(i, f_k)
				- \sum_{i \in \mathcal{B}} \sum_{j | (R_i - s^*_i) \cap s^*_j \ne \emptyset} \dc(i, j)
				+ \sum_{(i, j) | i \in \mathcal{B}, j \notin \mathcal{B}, f_k \in s^*_j} \dc(i, j)
				< 0 \nonumber \\
				&\sum_{i \in \mathcal{B}} \sum_{f_k' \in R_i - s^*_i} \cc(i, f_k')
				+ 2 \sum_{(i, j) | i \in \mathcal{B}, j \in \mathcal{B}, R_i \cap R_j = \emptyset}  \dc(i, j)
				+ 2 \sum_{(i, j) | i \in \mathcal{B}, j \notin \mathcal{B}, f_k \in s^*_j, R_i \cap s^*_j = \emptyset} \dc(i, j) \nonumber \\
				&< \sum_{i \in \mathcal{B}} \cc(i, f_k)
				+ \sum_{i \in \mathcal{B}} \sum_{ j | (R_i - s^*_i) \cap s^*_j \ne \emptyset} \dc(i, j) \label{paid-peering-multi-eq3}
			\end{align}
	
			The inequality above follows because $\{(i, j) | i \in \mathcal{B}, j \notin \mathcal{B}, f_k \in s^*_j, R_i \cap s^*_j = \emptyset\} \subseteq \{(i, j) | i \in \mathcal{B}, j \notin \mathcal{B}, f_k \in s^*_j\}$.
			
			Consider an assignment $s'$: start from $s^*$, let every agent $i$ in $\mathcal{B}$ switches to $\nBR_i(s^*, f_k)$. The total facility cost of $s'$ is no more than in $s^*$ because no agent would switch to a closed facility in $s^*$ by the definition of $\nBR_i(s^*)$. The total connection and disconnection cost in $s'$ compared to $s^*$ increases by at most the left hand side of Inequality~\ref{paid-peering-multi-eq3}, and decreases by the right hand side of it. Then the total social cost of $s'$ is less than $s^*$, which is a contradiction, so such subset $\mathcal{B}$ does not exist.
	\end{proof}

\end{document}